\documentclass[a4paper,11pt]{article}
\pdfoutput=1
\usepackage[utf8]{inputenc}
\usepackage[T1]{fontenc}
\usepackage{lmodern}
\usepackage[margin=1in]{geometry}
\usepackage{microtype}
\usepackage{amsmath,amssymb,amsthm,mathtools,mathrsfs}
\usepackage[dvipsnames]{xcolor}
\usepackage[numbers,sort&compress]{natbib}

\usepackage[hidelinks]{hyperref}
\usepackage{cleveref}
\usepackage{booktabs,tabularx,placeins}
\AtBeginDocument{\normalfont}
\hypersetup{
  pdftitle={From Heat to Homology: Spectral Gap Transfer for Exact Quantum Gibbs Sampling at All Temperatures},
  pdfauthor={Caesnan M. G. Leditto, Kuo-Chin Chen, Min-Hsiu Hsieh}
}
\newtheorem{theorem}{Theorem}[section]
\newtheorem{lemma}[theorem]{Lemma}
\newtheorem{proposition}[theorem]{Proposition}
\newtheorem{corollary}[theorem]{Corollary}
\newtheorem{definition}[theorem]{Definition}

\newtheorem{assumption}[theorem]{Assumption}
\theoremstyle{remark}
\newtheorem{remark}[theorem]{Remark}
\newtheorem*{addedremark}{Remark}
\crefname{assumption}{Assumption}{Assumptions}
\Crefname{assumption}{Assumption}{Assumptions}
\crefname{fact}{Fact}{Facts}
\Crefname{fact}{Fact}{Facts}
\Crefname{equation}{Eq.\!}{Eqs.\!}
\newcommand{\cE}{\mathcal E}
\newcommand{\cL}{\mathcal L}
\newcommand{\cT}{\mathcal T}

\newcommand{\cA}{\mathcal A}
\newcommand{\cK}{\mathcal K}
\newcommand{\cH}{\mathcal H}
\newcommand{\cB}{\mathcal B}
\newcommand{\cD}{\mathcal D}

\newcommand{\Tr}{\operatorname{Tr}}
\newcommand{\Var}{\operatorname{Var}}
\newcommand{\gap}{\operatorname{gap}}
\newcommand{\id}{\mathrm{id}}
\newcommand{\Fix}{\operatorname{Fix}}
\newcommand{\spec}{\operatorname{spec}}
\newcommand{\sfD}{\mathsf D}
\newcommand{\sfU}{\mathsf U}
\newcommand{\sfM}{\mathsf M}

\newcommand{\ad}{\operatorname{ad}}

\renewcommand{\Re}{\operatorname{Re}}
\newcommand{\Lab}{\mathfrak L}

\newcommand{\ket}[1]{|#1\rangle}
\newcommand{\bra}[1]{\langle#1|}

\title{From Heat to Homology: Spectral Gap Transfer for Exact Quantum Gibbs Sampling at All Temperatures}
\author{%
  Caesnan M. G. Leditto$^{1,3}$\thanks{Email: \href{mailto:caesnan.leditto@monash.edu}{\texttt{caesnan.leditto@monash.edu}}}
  \quad Kuo-Chin Chen$^{2}$ \quad Min-Hsiu Hsieh$^{2}$\\[0.8em]
  \small $^{1}$School of Physics and Astronomy, Monash University,\\
  \small Clayton, VIC 3168, Australia\\[0.3em]
  \small $^{2}$Hon-Hai Research Institute, Taipei, Taiwan\\[0.3em]
  \small $^{3}$Faculty of Science and Computer, Universitas Kristen Immanuel,\\
  \small Yogyakarta, Indonesia
}
\date{}
\begin{document}
\maketitle

\begin{abstract}
Preparing Gibbs states through dissipative dynamics requires controlling convergence for the chosen Hamiltonian $H$ and coupling operators. However, efficient implementation requires convergence guarantees from the spectral gap of the generator of the dynamics, which remains a crucial challenge. Recent results address the challenge by restricting the dynamics with structural assumptions about the Hamiltonian and its interactions. For the Chen–Kastoryano–Gilyen (CKG) construction, we prove an explicit comparison that converts a spectral gap estimate for the infinite-temperature generator into a lower bound on the spectral gap of the CKG quantum Gibbs sampler at every finite positive inverse temperature $\beta$. It allows a single spectral gap estimate at infinite temperature to certify convergence across Hamiltonians and temperatures. Such an estimate is often known from a classical spectral gap estimate, such as the classical random walk spectral gap. For Hodge Laplacians of simplicial complexes as Hamiltonians, we construct coupling operators whose generator at infinite temperature has the same spectral gap as a classical simplicial down--up walk.  A positive spectral gap of the walk and bounded spectral width $W=\lambda_{\max}(H)-\lambda_{\min}(H)$ give a generator spectral gap independent of the number of simplices at every fixed temperature. We then establish conditions under which the sampler approximately prepares zero-energy (harmonic) states in polynomial evolution time. This provides a conditional convergence guarantee for the state preparation in the thermal approach to quantum topological data analysis.
\end{abstract}

\section{Introduction}
\label{sec:introduction}

Simulating quantum systems to understand nature is a foundational objective
of quantum computing \cite{Feynman1982}.
Preparing thermal, or Gibbs, states enables quantum computers to estimate
finite-temperature energies and correlations
\cite{PoulinWocjan2009,Temme2011}.
Thermal-state preparation also provides a setting for studying quantum
computational advantage \cite{BergamaschiChenLiu2024,RajakumarWatson2026}
and universal quantum computation \cite{RouzeFrancaAlhambra2024},
with applications in optimization \cite{BrandaoSvore2017}, machine learning
\cite{AminEtAl2018}, and topological data analysis
\cite{ScaliUmeanoKyriienko2024}.

One approach to thermal-state preparation is to simulate dissipative
dynamics inspired by physical thermalization. These dynamics form a
quantum Markov semigroup that evolves an initial state toward the target
Gibbs state \cite{RWW2023,CKBG2023,DLL2025}.
For exact dissipative Gibbs sampling, efficient preparation requires
the target Gibbs state to be an exact stationary state, rapid
convergence from arbitrary initial states, and efficient access to the
Hamiltonian and coupling operators for simulating the dynamics.
The coupling operators specify the quantum updates used to construct the
dynamics. Recent constructions achieve exact stationarity and provide
simulation algorithms whose cost depends on the \emph{mixing time}, the
time needed to approach equilibrium to a chosen accuracy
\cite{CKGTechnical,DLL2025}.
For reversible dynamics satisfying quantum detailed balance and having
a unique stationary state, a lower bound on
the spectral gap of the Lindblad generator, which specifies the
continuous-time evolution, gives an upper bound on the mixing time.
This bound also depends on the desired accuracy and the stationary state
\cite{TemmeEtAl2010}.
The remaining challenge is to obtain a spectral gap bound strong enough
to guarantee rapid convergence at the temperature required by the
application.

Existing results obtain rapid convergence from locality or
interaction-strength bounds at high temperature
\cite{RouzeFrancaAlhambra2024,Bergamaschi}, one-dimensional structure
at every fixed finite temperature \cite{BergamaschiChen2025}, or
randomized couplings under temperature and spectral assumptions
\cite{RamkumarSoleimanifar}.
These approaches exploit structure that need not be available for a
chosen Hamiltonian and coupling family.
At infinite temperature, the stationary state is maximally mixed, and
the normalized infinite-temperature dynamics is determined by the original couplings alone.
Its convergence can sometimes be estimated using classical Markov
chains embedded in quantum channels \cite{AhiableEtAl2021}, or by
combining convergence estimates for independently evolving subsystems
\cite{KastoryanoReebWolf2012}.
Extending these estimates to finite temperature is difficult because
cooling can suppress transitions through higher-energy states that
connect low-energy states. The analysis must also control quantum
coherences, the off-diagonal entries of the density matrix in an energy
eigenbasis, as well as the diagonal populations \cite{Temme2013}.
We therefore ask \emph{whether a known convergence bound at infinite
temperature can provide a quantitative guarantee at finite temperature
for dynamics constructed from the same couplings, without locality or
energy-difference assumptions}.
Such a comparison would allow convergence estimates to be reused across
Hamiltonians and temperatures.

\paragraph{Main results.}
We address this question for the exact sampler of Chen, Kastoryano,
and Gily\'en (CKG) with its Metropolis transition weight
\cite{CKGTechnical}.
This construction uses energy filtering to modify the coupling operators
for the target temperature, which precise definition is given in
\Cref{sec:ckg}.
We fix the overall time scale through the normalization
$\sum_a A_a^\dagger A_a=I$, and require that each adjoint of the coupling operator
also belongs to the family, up to a phase.
Let $H$ be any finite-dimensional Hamiltonian,
$W=\lambda_{\max}(H)-\lambda_{\min}(H)$ its spectral width, and $\beta$ the inverse temperature.
We apply the comparison to Hodge Laplacians of simplicial complexes as
Hamiltonians \cite{HorakJost2013}. Their zero-energy, or harmonic,
subspaces encode topological information: the dimension of the harmonic
subspace is the corresponding Betti number, which counts connected
components or independent higher-dimensional holes
\cite{LloydGarneroneZanardi2016}.

\begin{enumerate}
\item \emph{Spectral gap comparison and polynomial evolution time.}
Let $\kappa_0>0$ be a spectral gap lower bound for the normalized
infinite-temperature generator
$\cL_{\cA,0}(X)=\sum_a A_a^\dagger X A_a-X$ on the full operator space.
We assume that its only fixed observables are scalar multiples of the
identity, so its stationary state is uniquely maximally mixed.
For the CKG generator $\cL_{\beta,\mathrm M}$, we prove
\begin{equation}
\gap(-\cL_{\beta,\mathrm M})
\ge\frac12e^{-\beta W}\kappa_0\qquad(\beta>0).
\label{eq:intro-ckg-gap}
\end{equation}
The comparison controls both populations and coherences and requires
no assumptions on locality or the spacing and multiplicity of energy
differences (\Cref{thm:metropolis-all-temperature}).
Standard mixing estimates \cite{TemmeEtAl2010} give polynomial
evolution time to inverse-polynomial trace-distance accuracy when
$\kappa_0$ is bounded below by an inverse polynomial,
$\beta W=O(\log n)$, and $\log D$ is polynomially bounded in the
instance size $n$, where $D$ is the Hilbert-space dimension.
The exponential dependence on $\beta W$ is necessary in general:
a three-dimensional system example has a gap asymptotic to
$\kappa_0e^{-\beta W}$ at low temperature
(\Cref{appth:thm:sharpness}).
Thus the coefficient in the exponent is optimal, although optimality
of the prefactor $1/2$ remains open.
A complementary perturbative estimate also covers the Gaussian
transition weight. With the exact spectral gap at infinite temperature as input,
its lower bound approaches the CKG Lindbladian spectral gap at infinite temperature as $\beta\to0$ for
fixed $H$ and couplings (\Cref{thm:comparison}).
Thus sufficiently small inverse temperatures preserve almost all of
the limiting gap; the permitted range may shrink with system size.

\item \emph{Approximate harmonic state preparation.}
For Hodge Laplacians, we construct couplings from the uniform simplicial
down--up walk, which moves between faces by removing a vertex and then
adding one. The associated normalized infinite-temperature generator
has exactly the walk's spectral gap (\Cref{thm:main}).
Established walk estimates, including those obtained from local
spectral expansion, therefore supply the input to our comparison
\cite{AlevLau2020,KaufmanOppenheim2020}.
Let $P_0$ project onto the harmonic subspace of $\Delta_k$,
$h=\dim\ker\Delta_k>0$, and $D_k$ be the number of $k$-faces.
The harmonic state is $\rho_{\rm harm}:=P_0/h$, the uniform mixture on this subspace.
For families of complexes with $n$ vertices, suppose the walk spectral gap and
dimension ratio $h/D_k$ are bounded below by inverse polynomials,
and the smallest positive eigenvalue $\delta_H$ of the chosen Hodge
Hamiltonian, with the same normalization used for $W$, satisfies
$\delta_H/W\ge\alpha>0$ with $\alpha$ independent of $n$.
A suitable $\beta W=c\log n$ then gives polynomial evolution time
to approximate the uniform mixture $P_0/h$ on the harmonic subspace
to inverse-polynomial trace-distance error
(\Cref{cor:harmonic-polynomial}).
A positive walk gap alone does not ensure rapid low-temperature
convergence. Our low-temperature upper bound and explicit unweighted
Hodge example show that the quantum gap can decay exponentially as
the inverse temperature increases while the walk gap stays positive
(\Cref{appth:thm:groundobstruction,appth:cor:barbell}).
\end{enumerate}

Our comparison allows a single spectral gap estimate at infinite
temperature to be reused across Hamiltonians and temperatures.
It gives an explicit bound for the prescribed CKG generator at every
finite temperature, with a loss determined by inverse temperature and
the Hamiltonian's spectral width. For Hodge Laplacians, the comparison transfers classical walk spectral gap estimates into convergence guarantees for thermal harmonic state preparation.
Thermal approaches to quantum topological data analysis (QTDA) seek to prepare the uniform mixture on a nonzero harmonic subspace by
cooling \cite{ScaliUmeanoKyriienko2024,SchmidhuberEtAl2025}.
Under the stated conditions, cooling suppresses positive-energy weight
while the dynamics still converges in polynomial evolution time.
This supplies a conditional convergence guarantee for the
state-preparation step in thermal QTDA.
Efficient implementation and extraction of Betti numbers require
separate analyses of access, simulation, and measurement costs.

\paragraph{Related work.}
Our technical contribution is the estimate that transfers the
Poincar\'e inequality at infinite temperature for the original couplings
to the finite-temperature filtered generator. This inequality
quantifies relaxation of all nonconstant observables through
$\kappa_0$.
Adjoint pairing and Plancherel's identity control the full operator
space without analyzing individual Bohr frequencies, the differences
between Hamiltonian eigenvalues, separately. Comparing Gibbs and
maximally mixed weights gives the explicit coefficient.

The factor $e^{-\beta W}$ already occurs in Li--Lu's spectral gap
comparison \cite{LiLu}, building on established comparison methods
\cite{JLR}.
Their result concerns generator pairs built from the same operators
with definite energy differences and satisfying
Gelfand--Naimark--Segal (GNS) detailed balance, a stronger form of
quantum reversibility than the Kubo--Martin--Schwinger (KMS)
detailed balance used by CKG.
The paired maximally mixed generator in that comparison need not
coincide with the generator constructed from the original unfiltered
couplings. Our additional estimate applies to the prescribed
Gaussian-filtered CKG dynamics, which need not satisfy GNS detailed
balance. General comparison criteria \cite{JW} do not by themselves
supply the explicit coefficient for this generator.

Basso et al. \cite{BassoEtAl2025} transfer a gap bound for observables
commuting with the Hamiltonian to the full quantum gap at the same
temperature for Davies dynamics. The loss depends on arithmetic
progressions of energies.
Our input instead controls the full operator space at infinite
temperature, and our estimate quantifies the loss in the gap guarantee at finite
temperature without assumptions on energy differences. 

Table~\ref{tab:comparison} presents the concise comparison with other approaches. Further comparisons with other recent CKG Lindbladian spectral gap results appear in \Cref{app:comparison-context}.
\begin{table}[t]
\centering
\small
\setlength{\tabcolsep}{4pt}
\renewcommand{\arraystretch}{1.08}
\begin{tabularx}{\textwidth}{@{}
>{\raggedright\arraybackslash}p{0.27\textwidth}
>{\raggedright\arraybackslash}p{0.28\textwidth}
>{\raggedright\arraybackslash}X@{}}
\toprule
\textbf{Comparison and reference gap}
& \textbf{Gap lower bound}
& \textbf{Assumptions and limitation} \\
\midrule
Present work: original full-operator gap estimate $\kappa_0$
at infinite temperature
& $\tfrac12e^{-\beta W}\kappa_0$
& CKG dynamics with normalized couplings closed under adjoints up to phases; arbitrary $H$ and $\beta>0$.
The bound can be exponentially weak in $\beta W$. \\
\addlinespace
Li--Lu \cite{LiLu}: paired maximally mixed generator gap
$g_{\mathrm{pair}}$
& $e^{-\beta W}g_{\mathrm{pair}}$
& Specified primitive GNS-symmetric generator pair sharing modular eigenoperators.
The paired generator differs from the infinite temperature
generator. \\
\addlinespace
Basso et al. \cite{BassoEtAl2025}: minimum classical gap
$g_{\mathrm{cl}}^{\min}(\beta)$ at the same temperature
& $\frac{g_{\mathrm{cl}}^{\min}(\beta)}{2\ell}$
& Davies dynamics.
Transfers control from observables commuting with $H$
to the full operator space; the loss depends on
arithmetic progressions of energies. \\
\bottomrule
\end{tabularx}
\caption{Spectral gap comparisons and their required reference
estimates. The paired gap $g_{\mathrm{pair}}$ may depend on temperature.
The quantity $g_{\mathrm{cl}}^{\min}(\beta)$ is the minimum classical
gap over energy eigenbases, including choices within degenerate
eigenspaces. The parameter $\ell$ bounds the length of proper
arithmetic progressions of distinct energies.
Each result uses its own coupling and rate conventions.
}
\label{tab:comparison}
\end{table}

\paragraph{Organization.}
\Cref{sec:ckg} defines the CKG dynamics and spectral gap conventions.
\Cref{sec:comparison-main} presents the comparison, its limitations,
and its mixing-time consequence.
\Cref{sec:hodge-application} applies it to Hodge Laplacians.
The appendices contain the proofs, the perturbative stability
estimate, the low-temperature examples, and further literature
comparisons.

\FloatBarrier

\section{Exact quantum Gibbs sampler}
\label{sec:ckg}

The CKG construction combines Gaussian operator Fourier filtering with
a coherent correction to obtain a quantum Markov semigroup satisfying
KMS detailed balance \cite{CKGTechnical}. We use this construction to
fix the stationary Gibbs state and study convergence for a prescribed
family of coupling operators. Our comparison starts from a
Poincar\'e inequality for the normalized infinite-temperature
generator, whose stationary state is maximally mixed. Its spectral
gap can often be obtained from classical Markov-chain spectral gap estimates or
single-site calculations. The normalization used in the comparison
fixes the overall time scale.

Throughout, $\cB(\cH)$ denotes the linear operators on a
finite-dimensional Hilbert space $\cH$, $I$ is the identity,
$X^\dagger$ is the adjoint, and
$[X,Y]=XY-YX$, $\{X,Y\}=XY+YX$.

\begin{assumption}[Hamiltonian and coupling operators]
\label{ass:ckg}
$\cH$ has dimension $D\ge2$;
$H=H^\dagger\in\cB(\cH)$ has spectral width
$W:=\lambda_{\max}(H)-\lambda_{\min}(H)$; and
$\{A_a\}_{a\in\mathcal A}\subset\cB(\cH)$ is a finite family of
\emph{coupling operators} satisfying
\begin{equation}
 S_{\cA}:=\sum_{a\in\mathcal A}A_a^\dagger A_a\preceq I,
 \qquad
 A_{\iota(a)}=e^{i\varphi_a}A_a^\dagger
 \label{eq:generic-CKG-assumptions}
\end{equation}
for an involution $\iota$ of $\mathcal A$ and real phases
$\varphi_a$. Here $S_{\cA}\preceq I$ means that $I-S_{\cA}$ is
positive semidefinite.
\end{assumption}

The normalization bound and adjoint closure are the hypotheses of the
CKG construction; phases can be removed without changing the generator
(\Cref{lem:relabelling}). For the spectral gap comparison, we additionally
require $S_{\cA}=I$ and a Poincar\'e inequality. This inequality bounds the squared commutator sum in
\eqref{eq:generic-L0-Poincare} below by a positive multiple of
the observable's variance. At infinite temperature, the variance is the squared
distance from the scalar operators in the normalized Hilbert--Schmidt
norm. A positive Poincar\'e constant therefore excludes additional
fixed observables, including off-diagonal ones.

\begin{assumption}[Poincar\'e inequality at infinite temperature]
\label{ass:poincare}
\Cref{ass:ckg} holds; moreover $S_{\cA}=I$ and there is a constant
$\kappa_0>0$ with
\begin{equation}
 \frac12\sum_{a\in\mathcal A}\|[A_a,O]\|_\tau^2
 \ge\kappa_0\Var_\tau(O)
 \qquad\text{for every }O\in\cB(\cH),
 \label{eq:generic-L0-Poincare}
\end{equation}
where $\tau:=I/D$,
$\langle X,Y\rangle_\tau:=D^{-1}\Tr(X^\dagger Y)$,
$\|X\|_\tau^2:=\langle X,X\rangle_\tau$, and
$\Var_\tau(X):=\inf_{c\in\mathbb C}\|X-cI\|_\tau^2$.
\end{assumption}

Under \Cref{ass:poincare}, the normalized infinite-temperature
generator is
\[
 \cL_{\mathcal A,0}(O):=\sum_a A_a^\dagger O A_a-O.
\]
Its Dirichlet form satisfies
$-\langle O,\cL_{\mathcal A,0}(O)\rangle_\tau
=\frac12\sum_a\|[A_a,O]\|_\tau^2$.
Thus $\kappa_0$ is a lower bound on its spectral gap, and the largest
admissible $\kappa_0$ is the gap itself. The inequality must hold on
the entire operator space; a gap estimate for a classical restriction
alone does not establish it.

A classical Markov chain nevertheless provides one direct way to
verify this inequality through a measure-and-prepare channel. In the
Schr\"odinger picture, this channel measures in a fixed basis, applies
a classical transition to the outcome, and prepares the resulting
basis state. Let $Q$ be an irreducible symmetric stochastic matrix on
$\{1,\ldots,D\}$, with eigenvalues
$1=\lambda_1(Q)>\lambda_2(Q)\ge\cdots\ge\lambda_D(Q)$.
In an orthonormal basis $\{|f\rangle\}_{f=1}^D$, the couplings
$K_{fg}:=\sqrt{Q(f,g)}|g\rangle\langle f|$ define its
Heisenberg-picture channel
\[
 \Phi_Q(O):=\sum_{f,g}K_{fg}^\dagger O K_{fg}
 =\sum_f\left(\sum_g Q(f,g)\langle g|O|g\rangle\right)
 |f\rangle\langle f|.
\]
The generator $\Phi_Q-\id$ acts as $Q-I$ on diagonal observables
and as $-\id$ on off-diagonal operators. Its spectral gap is therefore
$\min\{1,1-\lambda_2(Q)\}$, which supplies a valid
$\kappa_0$ in \Cref{ass:poincare}. For this particular construction,
the classical gap thus gives a gap bound on the entire operator space.
Other quantum channels can induce the same classical walk on diagonal
operators.
The construction and proof are recalled in
\Cref{def:Markov-kernel-embedding,thm:associated-gap}.

We now define the finite-temperature CKG dynamics for a general
coupling family satisfying \Cref{ass:ckg}. Fix an inverse temperature
$\beta>0$ and put
$\sigma_\beta:=e^{-\beta H}/\Tr(e^{-\beta H})$.
Operator Fourier filtering resolves the energy changes induced by
each coupling with finite Gaussian energy resolution. The filtered
generator retains terms coupling distinct Bohr frequencies, and a
coherent correction ensures exact KMS detailed balance.
We use the normalized Gaussian time filter and shifted Metropolis
transition weight of \cite{CKGTechnical}:
\begin{gather*}
 f_\beta(t):=\frac{e^{-t^2/\beta^2}}
 {\sqrt{\beta\sqrt{\pi/2}}},\qquad
 \gamma_{\rm M}(\omega):=
 e^{-\beta\max\{\omega+1/(2\beta),0\}},\\
 \widehat A_a(\omega):=\frac1{\sqrt{2\pi}}
 \int_{\mathbb R}f_\beta(t)e^{-i\omega t}
 e^{iHt}A_ae^{-iHt}\,dt.
\end{gather*}
In particular, $\int_{\mathbb R}f_\beta(t)^2dt=1$.
Write $\gamma=\gamma_{\rm M}$ and
$H=\sum_E E P_E$, where $P_E$ are its spectral projections.
The loss operator and coherent correction are
\begin{align*}
 R_\gamma&:=\sum_a\int_{\mathbb R}\gamma(\omega)
 \widehat A_a(\omega)^\dagger\widehat A_a(\omega)\,d\omega,\\
 B_\gamma&:=\frac{i}{2}\sum_{E,E'}
 \tanh\!\left(\frac{\beta(E-E')}{4}\right)P_E R_\gamma P_{E'}.
\end{align*}
These operators define the exact CKG generator in the
Schr\"odinger picture:
\begin{align}
    \cL_{\gamma,*}(\rho)
 =-i[B_\gamma,\rho]
 +\sum_a\int_{\mathbb R}\gamma(\omega)
 \widehat A_a(\omega)\rho\widehat A_a(\omega)^\dagger\,d\omega
 -\frac12\{R_\gamma,\rho\}.\label{eq:exact-CKG-generator}
\end{align}
Its Heisenberg-picture adjoint $\cL_\gamma$ is determined by
$\Tr[\cL_{\gamma,*}(\rho)O]=\Tr[\rho\cL_\gamma(O)]$.
The notation suppresses its dependence on $\beta$, $H$, and the
coupling family. For $\gamma=\gamma_{\rm M}$, we also write
$\cL_{\beta,\mathrm M}:=\cL_{\gamma_{\rm M}}$. The correction $B_\gamma$ enforces KMS detailed
balance for these filtered jumps
\cite[Theorem~I.1 and Lemma~II.1]{CKGTechnical}.
\Cref{app:ckg-details} provides the convention checks and the additional
Gaussian transition weight used in the perturbative comparison.
Finite energy resolution avoids the need to resolve individual Bohr
frequencies, but does not by itself guarantee a large spectral gap.

To characterize convergence, we use the Gibbs-weighted inner product
in which the CKG generator is self-adjoint. For a full-rank state
$\sigma$, the KMS inner product is
$\langle X,Y\rangle_{\sigma,1/2}
:=\Tr(\sigma^{1/2}X^\dagger\sigma^{1/2}Y)$.
CKG detailed balance means that $\cL_\gamma$ is self-adjoint for this
inner product with $\sigma=\sigma_\beta$; in particular,
$\cL_{\gamma,*}(\sigma_\beta)=0$
\cite[Theorem~I.1]{CKGTechnical}.
This symmetry gives the variational characterization
\[
 \gap(-\cL_\gamma):=
 \inf_{\substack{X\ne0\\\Tr(\sigma_\beta X)=0}}
 \frac{-\langle X,\cL_\gamma(X)\rangle_{\sigma_\beta,1/2}}
 {\langle X,X\rangle_{\sigma_\beta,1/2}},
\]
which we use to bound the slowest nonstationary decay rate.
A positive gap implies that the only fixed observables are scalar
multiples of $I$. In this finite-dimensional KMS-symmetric setting,
the dynamics is then primitive: $\sigma_\beta$ is its unique stationary
state and every initial state converges to it. These 
consequences are recalled in \Cref{def:gap,lem:gap-primitive}.

To compare this spectral gap with the bound $\kappa_0$, we
distinguish the normalized infinite-temperature generator from the
infinite-temperature limit of the CKG generator. Under
\Cref{ass:poincare},
$\lim_{\beta\downarrow0}\cL_{\gamma_{\rm M}}
=q\cL_{\mathcal A,0}$, where
\[
 q:=\operatorname{erfc}\!\left(\frac1{2\sqrt2}\right)
 =\frac2{\sqrt\pi}\int_{1/(2\sqrt2)}^\infty e^{-u^2}\,du
 \simeq0.617.
\]
Throughout, $\kappa_0$ denotes a spectral gap lower bound for
$\cL_{\mathcal A,0}$ with the normalization specified above.
The spectral gap of the CKG generator at infinite temperature is $q$ times that
of $\cL_{\mathcal A,0}$.

\section{Spectral gap comparison at finite temperature}
\label{sec:comparison-main}

We use a spectral gap comparison to turn a Poincar\'e inequality for
the normalized infinite-temperature generator
$\cL_{\mathcal A,0}$ into a convergence bound for the finite-temperature
CKG dynamics. The Poincar\'{e} constant $\kappa_0$ can be obtained, for
example, from a classical Markov-chain estimate or tensorization.
Keeping the coupling family fixed, we compare squared commutator norms
before and after Gaussian Fourier filtering. This yields a bound without requiring a
direct spectral analysis of the generator at finite temperature.

To give an overview of the proof, we use the Heisenberg-picture adjoint of \Cref{eq:exact-CKG-generator} for an observable $X=X^\dagger$. We compare the commutator $[\widehat A_a(\omega),X]$ with the corresponding commutator for the adjoint coupling at frequency $-\omega$. Closure of the coupling family under adjoints implies that these commutators are adjoints of one another up to a phase. Averaging their contributions and using $\gamma_{\rm M}(\omega)+\gamma_{\rm M}(-\omega)\ge1$ gives a lower bound in terms of squared Hilbert--Schmidt norms. Using Plancherel's identity \cite[Theorem~9.13]{Rudin1987}, the Poincar\'{e} inequality bounds the frequency integral because simultaneous unitary conjugation of the couplings preserves the Poincar\'{e} constant. Comparing $-\operatorname{Re}\operatorname{Tr}\left(\sigma_\beta X\cL_{\beta,\mathrm M}(X)\right)$ and $\left[\operatorname{Tr}(\sigma_\beta X^2)
-\operatorname{Tr}(\sigma_\beta X)^2\right]$ using the smallest and largest eigenvalues of $\sigma_\beta$ produces the ratio $\lambda_{\min}(\sigma_\beta)/\lambda_{\max}(\sigma_\beta)=e^{-\beta W}$. Finally, KMS symmetry \cite[Theorem~I.1]{CKGTechnical} ensures that the generator has real eigenvalues \cite[Sec.~II]{KastoryanoTemme2013}. Since the generator preserves adjoints, for each nonzero eigenvalue $-\lambda$ we can choose a nonzero operator $X=X^\dagger$ satisfying $\cL_{\beta,\mathrm M}(X)=-\lambda X$. Substituting this equation into $-\operatorname{Re}\operatorname{Tr}\left(\sigma_\beta X\cL_{\beta,\mathrm M}(X)\right)
\ge \frac12e^{-\beta W}\kappa_0
\left[\operatorname{Tr}(\sigma_\beta X^2)
-\operatorname{Tr}(\sigma_\beta X)^2\right]$ gives $\lambda\ge\frac12e^{-\beta W}\kappa_0$. No assumption on the separation of Bohr frequencies is needed. The full proof is in \Cref{app:comparison-proof}.

\begin{theorem}[Spectral gap comparison at every finite inverse temperature]
\label{thm:metropolis-all-temperature}
Adopt \Cref{ass:poincare}, fix $\beta>0$, and use the transition weight
$\gamma_{\rm M}$ defined in \Cref{sec:ckg}. With $b:=\beta W$, the
exact CKG generator is primitive, $\sigma_\beta$ is its unique stationary
state, and
\begin{equation}
 \gap(-\cL_{\gamma_{\rm M}})\ge\frac12e^{-b}\kappa_0.
 \label{eq:metropolis-all-temperature-gap}
\end{equation}
No upper bound on $\beta W$ and no assumption on the number, spacing, or
multiplicity of the Bohr frequencies are required.
\end{theorem}

The factor $e^{-\beta W}$ is the ratio of the smallest and largest
eigenvalues of the Gibbs state. Comparisons through this ratio are
established under Gelfand--Naimark--Segal (GNS) detailed balance,
which means self-adjointness with respect to the inner product
$\Tr(\sigma X^\dagger Y)$ \cite{JLR,LiLu}. The theorem above evaluates
the comparison between the Gaussian-filtered CKG generator and the
normalized infinite-temperature generator constructed from the same
couplings. It requires KMS symmetry but does not assume GNS symmetry.
Indeed, \Cref{prop:non-GNS-example} gives a qubit example covered by
the theorem but outside the hypotheses of the cited GNS comparisons.
General comparison principles also exist without GNS symmetry
\cite{JW}, and qualitative preservation of ergodicity under
operator-Fourier filtering is already known
\cite[Theorem~13]{GCDK2024}. The contribution here is the explicit
spectral gap lower bound for the stated generator pair, including its
dependence on temperature.

The exponential dependence on $\beta W$ is necessary under these
assumptions. \Cref{appth:thm:sharpness} constructs a three-dimensional system chain
whose two zero-energy states communicate only through a state of
energy $W$. Its associated normalized infinite-temperature
generator has spectral gap $\kappa_0=1/4$, while the spectral gap of its CKG generator
is asymptotic to $\kappa_0e^{-\beta W}$ as $\beta W\to\infty$.
Consequently, no bound $c\kappa_0e^{-\eta\beta W}$ with universal
$c>0$ and $\eta<1$ can hold under the theorem's hypotheses.
The coefficient of $\beta W$ in the exponent is therefore optimal;
the optimal numerical prefactor remains undetermined.
A complementary estimate near infinite temperature, also covering
the Gaussian transition weight, is given in \Cref{thm:comparison}.

The spectral gap bound also controls the time required for the CKG
dynamics to approach the Gibbs state. The inverse spectral gap is
the relaxation time. To obtain a mixing-time bound valid for every
initial state, we use the quantum $\chi^2$ estimate
\cite[Lemmas~5 and~12]{TemmeEtAl2010}. This estimate converts
exponential decay in a weighted Hilbert--Schmidt norm into
trace-distance convergence, with a prefactor depending on the
smallest Gibbs eigenvalue. For $\gamma=\gamma_{\rm M}$, let
$t_{\rm mix}(\epsilon)$ be the least $t$ such that
$\tfrac12\|e^{t\cL_{\gamma,*}}(\rho)-\sigma_\beta\|_1\le\epsilon$
for every state $\rho$, where
$\|X\|_1=\Tr\sqrt{X^\dagger X}$.
Using $\lambda_{\min}(\sigma_\beta)\ge e^{-\beta W}/D$ and writing
$[x]_+:=\max\{x,0\}$ gives
\[
 t_{\rm mix}(\epsilon)
 \le\frac{2e^{\beta W}}{\kappa_0}
 \left[\log\frac1{2\epsilon}
       +\frac12\log(De^{\beta W}-1)\right]_+
 \qquad(0<\epsilon<1).
\]
The derivation is recalled in \Cref{sec:mixing}.

\begin{addedremark}[Polynomial mixing]
Consequently, for a family of instances indexed by $n$, the comparison
guarantees polynomial semigroup mixing time at fixed accuracy whenever
$\kappa_0$ is bounded below by an inverse polynomial,
$\beta W=O(\log n)$, and $\log D$ is polynomially bounded.
Converting this evolution time into a circuit cost additionally
requires efficient access to $H$ and the coupling operators, together
with control of the temperature-dependent cost of CKG simulation.
\Cref{appth:cor:log-temperature} gives the explicit estimate, using
\cite[Lemmas~5 and~12]{TemmeEtAl2010}.
\end{addedremark}

\begin{addedremark}[Low-energy concentration]
Under \Cref{appth:cor:log-temperature}, consider a Hamiltonian with exactly two equally large
eigenspaces at energies $0$ and $W>0$. For their ground-space
projection $P_0$ and $\beta W=c\log n$ with $c>0$,
$\Tr(P_0\sigma_\beta)=(1+n^{-c})^{-1}\to1$, whereas
$\Tr(P_0I/D)=1/2$. Polynomial mixing can therefore coexist with
substantial low-energy concentration.
\end{addedremark}

The dependence on the global spectral width distinguishes this
guarantee from local high-temperature bounds
\cite{RouzeFrancaAlhambra2024,Bergamaschi}. Those results exploit local
interaction bounds and can give a gap $\Omega(1/m)$ for an $m$-site
system under the same overall generator normalization. Our comparison
applies to any coupling family satisfying \Cref{ass:poincare} and
every $\beta>0$. However, when $W=\Theta(m)$,
$\kappa_0=\Theta(1/m)$, and $\beta>0$ is fixed, the lower bound it
provides is $e^{-\Theta(m)}/m$. Thus the general comparison can give a
much weaker guarantee than local high-temperature results for
extensive Hamiltonians. Detailed comparisons are in \Cref{app:comparison-context}. 

\section{Gibbs sampling for simplicial Hodge Laplacians}
\label{sec:hodge-application}

We now construct coupling operators from a classical down--up walk
on simplices and use \Cref{thm:metropolis-all-temperature} to
bound convergence to the Gibbs state of a Hodge Laplacian. The
classical walk gap determines the spectral gap of the associated
normalized infinite-temperature generator. A separate spectral
quantity, the smallest positive eigenvalue of the chosen Hodge Hamiltonian,
together with the dimensions of the harmonic and positive-energy
subspaces, determines how low the temperature must be for the Gibbs state to concentrate in the
harmonic subspace. This distinction is relevant when Gibbs sampling
is used to prepare the harmonic state for quantum topological data
analysis \cite{LloydGarneroneZanardi2016,McArdleGilyenBerta2022}.

Let $X\subseteq2^{[n]}$ be a finite unweighted simplicial complex
containing every vertex, and let $X(j)$ denote its set of
$j$-dimensional faces. Fix $1\le k\le\dim X$, and write
$F=X(k)$ and $D_k=|F|\ge2$. The complex $X$ need not be pure.
Orient faces by increasing vertex order, and let
$\partial_j:\ell^2(X(j))\to\ell^2(X(j-1))$ be the oriented
simplicial boundary. On $\cH_k=\ell^2(F)$, use the unweighted
Hodge Laplacian
$\Delta_k=\partial_k^\dagger\partial_k+
\partial_{k+1}\partial_{k+1}^\dagger$, with
$\partial_{k+1}=0$ when $X(k+1)=\emptyset$
\cite{Eckmann1944,HorakJost2013}.
Choose $\alpha_j\ge\|\partial_j\|$ for $j=k,k+1$, where
$\|\cdot\|$ is the operator norm, and set
\[
 \Lambda_k=\alpha_k^2+\alpha_{k+1}^2,
 \qquad H_{X,k}=\Delta_k/\Lambda_k,
 \qquad W=\lambda_{\max}(H_{X,k})-\lambda_{\min}(H_{X,k}).
\]
Then $\Lambda_k>0$, $0\preceq H_{X,k}\preceq I$, and $W\le1$.
The Gibbs state is
$\sigma_\beta=e^{-\beta H_{X,k}}/\Tr(e^{-\beta H_{X,k}})$;
inverse temperature $\theta$ for $\Delta_k$ corresponds to
$\beta=\theta\Lambda_k$.

To construct the coupling operators, let $Y=\langle F\rangle$
be the pure $k$-dimensional complex generated by $F$.
The uniform down--up walk $P_u$ on $F$ deletes a uniformly chosen
vertex from the current face and then chooses uniformly among the
$k$-faces containing the remaining $(k-1)$-face, including the
original face. This walk is symmetric and positive semidefinite,
with uniform stationary distribution. Assume it is irreducible and
write $\kappa_u=1-\lambda_2(P_u)>0$.
Its associated measure-and-prepare channel, in the Heisenberg
picture, is
$\Phi_{P_u}(A)=\sum_{f,g\in F}P_u(f,g)\langle g|A|g\rangle
|f\rangle\langle f|$.
For example,
$\{\sqrt{P_u(f,g)}|g\rangle\langle f|\}_{f,g\in F}$
is a Kraus family closed under adjoints.

\begin{theorem}[CKG spectral gap bound from a simplicial down--up walk]
\label{thm:main}
Under the preceding assumptions, construct the exact CKG generator
for $H_{X,k}$ from any Kraus family for $\Phi_{P_u}$ closed under
adjoints up to phases. For transition weight $\gamma_{\rm M}$ and every
$\beta>0$, this generator is primitive with unique stationary state
$\sigma_\beta$, and
\begin{equation}
 \gap(-\cL_{\gamma_{\rm M}})
 \ge\frac12e^{-\beta W}\kappa_u.
 \label{eq:main-all-temperature-gap}
\end{equation}
\end{theorem}

The theorem follows by identifying the spectral gap of the
normalized infinite-temperature generator $\Phi_{P_u}-\id$.
The channel $\Phi_{P_u}$ acts as $P_u$ on diagonal observables and
annihilates off-diagonal operators. Since $P_u$ is symmetric and
positive semidefinite, \Cref{thm:associated-gap} gives the exact
spectral gap $\kappa_u$ for $\Phi_{P_u}-\id$.
The general comparison theorem then yields
\eqref{eq:main-all-temperature-gap}. Changing the Kraus
representation of this fixed channel leaves the CKG generator
unchanged (\Cref{lem:relabelling}). The full proof is in
\Cref{app:hodge-details}.

The spectral gap bound also gives a trace-distance mixing estimate.
By \Cref{lem:mixing},
\[
 t_{\rm mix}(\epsilon)
 =O\!\left(e^{\beta W}\kappa_u^{-1}
 [\log D_k+\beta W+\log(1/\epsilon)]\right),
 \qquad 0<\epsilon<1.
\]
At fixed accuracy, a size-independent positive lower bound on
$\kappa_u$ and a size-independent upper bound on $\beta W$
therefore give $t_{\rm mix}=O(\log D_k)$.
These estimates require a bound on the spectral width of the Hodge
Hamiltonian, but do not require its diagonalization.

Quantitative lower bounds on the classical walk gap can be
substituted directly into these estimates. In particular, the
Alev--Lau theorem \cite{AlevLau2020} bounds the down--up walk gap
using the spectra of walks on links of faces. This local-to-global
theorem concerns the classical walk; combining it with
\Cref{thm:main} gives a bound on the spectral gap of the CKG generator.
Other lower bounds on $\kappa_u$ can be used in the same way.
The induced link weights and the resulting product bound are
given in \Cref{app:hodge-details,app:hodge-support}.

Convergence to a Gibbs state and concentration in the harmonic subspace
can be controlled together. Let $0<h:=\dim\ker\Delta_k<D_k$, let $P_0$
project onto this kernel, and put $R=(D_k-h)/h$.
For the smallest positive eigenvalue $\delta_H$ of $H_{X,k}$,
\[
 \frac12\|\sigma_\beta-P_0/h\|_1\le R e^{-\beta\delta_H}.
\]
Combining this estimate with the mixing bound gives explicit sufficient
conditions for preparing the harmonic state
(\Cref{cor:harmonic-approximation,cor:harmonic-polynomial}).
For families on $n$ vertices, suppose $\kappa_u\ge n^{-a}$,
$R\le n^r$, and $\delta_H/W\ge\alpha>0$, with $a,r,\alpha$ independent
of $n$. For fixed $s>0$, choose a constant $c\ge(r+s)/\alpha$ and
$\beta W=c\log n$. The Gibbs-state error relative to the harmonic state is then at most $n^{-s}$;
running the dynamics to Gibbs-state error $n^{-s}$ gives trace-distance
error at most $2n^{-s}$ from $P_0/h$ in polynomial evolution time.
Here $\log D_k\le n\log2$.

This supplies a conditional convergence guarantee for thermal quantum
topological data analysis \cite{ScaliUmeanoKyriienko2024}. A positive walk
gap alone does not suffice. Ideal projection of $I/D_k$ onto
$\ker\Delta_k$ succeeds with probability $h/D_k=1/(1+R)$, which is at
least inverse polynomial under $R\le n^r$; no advantage
over other preparation methods is established \cite{ChowdhurySomma2017}.
Boundary-norm estimates and further CKG spectral gap bounds are proved in
\Cref{app:hodge-details,app:hodge-support}.

\section{Discussion}
\label{sec:conclusion}

We proved the comparison
$\gap(-\cL_{\beta,\mathrm M})\ge\tfrac12e^{-\beta W}\kappa_0$
for the exact CKG generator with normalized coupling
operators closed under adjoints up to phases. Here $\kappa_0$ is a
Poincar\'{e} constant of the normalized infinite-temperature generator
constructed from the same couplings. The bound controls both populations
and coherences and requires no assumptions on locality or the spacing
and multiplicity of energy differences. It allows a single spectral
gap estimate at infinite temperature to provide convergence guarantees
across Hamiltonians and temperatures.
Appendix~\ref{app:comparison-context} distinguishes this result from
GNS-symmetric comparisons \cite{JLR,LiLu} and general comparison
principles \cite{JW}; our qubit examples show that the CKG generator
need not be GNS symmetric. The perturbative comparison
(\Cref{thm:comparison}) also covers the Gaussian transition weight.
When $\kappa_0$ is the actual spectral gap of the normalized
infinite-temperature generator, its lower bound approaches
$q_\gamma\kappa_0$, the CKG spectral gap at infinite temperature,
as $\beta W\to0$. This gives one-sided spectral-gap stability for
fixed couplings, although the guaranteed inverse-temperature range
can shrink with system size.

At fixed accuracy, the mixing-time bound is polynomial for instance families indexed by $n$ when $\kappa_0$ is bounded below by an inverse
polynomial, $\beta W=O(\log n)$, and $\log D$ is polynomially bounded.
This regime can still produce substantial low-energy concentration.
The guarantee can become exponentially weak when the spectral width
grows with system size. Our three-dimensional system example in \Cref{app:optimal} shows that the
exponential dependence on $\beta W$ is necessary under the stated
assumptions \Cref{appth:thm:sharpness} and \Cref{appth:cor:necessity}, although a small lower bound alone does not establish
slow convergence for a particular instance. Stronger guarantees
therefore require additional assumptions. A main direction is to
identify structural conditions on Hamiltonians and coupling operators
that improve low-temperature bounds, including for noncommuting
interactions and temperatures that decrease with system size.
Local CKG Lindbladian spectral gap bounds
\cite{RouzeFrancaAlhambra2024,Bergamaschi,BergamaschiChen2025}
suggest exploiting interaction structure; energy-barrier bounds for
Davies generators of commuting Pauli Hamiltonians \cite{Temme2017}
motivate seeking analogous CKG bounds controlled by transitions
allowed by the couplings.

The measure-and-prepare quantum channel construction transfers
classical walk spectral gap estimates to this quantum convergence
analysis. For Hodge Laplacians acting as Hamiltonians, the associated
infinite-temperature generator has the same gap as the
uniform simplicial down--up walk. Local spectral expansion
\cite{AlevLau2020} supplies quantitative bounds. In particular, for
the bounded-degree families in
\Cref{cor:bounded-degree,cor:HDX-mixing}, the CKG dynamics with
transition weight $\gamma_{\mathrm M}$ has a positive spectral gap
bound independent of the number of faces at each fixed inverse
temperature with respect to the unscaled Hodge Laplacian $\Delta_k$
and mixes in $O(\log D_k)$ time at fixed accuracy.
Combining the general walk bound with Gibbs-state concentration gives
polynomial evolution time to approximate the harmonic state when the
walk spectral gap and dimension ratio $h/D_k$ are bounded below by
inverse polynomials and the smallest positive eigenvalue of the Hodge
Hamiltonian is bounded below by a constant fraction of the spectral
width (\Cref{cor:harmonic-polynomial}). This establishes a conditional
convergence guarantee for the state preparation step in thermal
quantum topological data analysis \cite{ScaliUmeanoKyriienko2024}.

Our low-temperature upper bound and unweighted Hodge Laplacian
example show that a spectral gap of the walk alone does not ensure
rapid thermalization to prepare the harmonic state. Reducible projected couplings
obstruct a uniform positive gap as $\beta\to\infty$ for fixed
Hamiltonian and couplings. Irreducibility of the projected couplings
alone has not been proved sufficient; relaxation from positive-energy
states must also be controlled. We therefore aim to find verifiable
sufficient conditions for rapid harmonic state preparation across
broader families of simplicial complexes, controlling both dynamics
within the harmonic subspace and relaxation from positive-energy
states. This requires understanding how coupling choices affect
finite-temperature convergence beyond what is captured by the
classical walk and infinite-temperature gap. Combining coupling
families can improve the infinite-temperature bound
(\Cref{subsec:mixture}), whereas changing a Kraus representation leaves
the CKG generator unchanged (\Cref{lem:relabelling}). Random coupling
constructions and their dependence on the low-energy spectral
distribution \cite{RamkumarSoleimanifar} suggest further approaches
to coupling design, as discussed in
Appendix~\ref{app:comparison-context}.

A further objective is to develop complete thermal QTDA algorithms
by combining convergence guarantees with efficient implementation,
initialization, and measurement procedures. Efficient Hamiltonian
and coupling access, initialization on valid faces, and controlled
simulation errors remain necessary
\cite[Theorem~I.2]{CKGTechnical}. Extracting topological information
requires a separate analysis of estimation error and sample complexity.
The harmonic state has purity $1/h$, and independent SWAP-test
estimation of $h$ to fixed relative accuracy uses $O(h^2)$ repetitions
at fixed success probability
\cite[Sec.~IV]{ScaliUmeanoKyriienko2024}. This procedure can be
exponentially costly despite polynomial evolution time.
Projector-rank estimation using quantum singular value transformation
and amplitude estimation offers another approach
\cite[Sec.~4]{McArdleGilyenBerta2022}, but requires unitary access to
the required block encodings and their inverses beyond a mixing-time
bound. Hodge Hamiltonians arising in Khovanov homology provide another
possible application \cite{SchmidhuberEtAl2025}; appropriate couplings
and spectral gap estimates at infinite temperature would need to be established
there. The present results certify the specified dynamics, without
establishing an advantage over other preparation methods or a fully
efficient algorithm for Betti-number estimation. Bounding the total
query, gate, and sample costs of complete thermal QTDA algorithms
would determine when these preparation guarantees yield efficient
extraction of topological information.

\section*{Acknowledgement}
We thank Angus Southwell and Ryu Hayakawa for discussing the implications of the main results and future directions during the preparation of the manuscript.  

\section*{AI Statement}
OpenAI GPT 5.6 and GPT-6 Astra were used to expand the original research ideas and hypotheses, prove the main theorems, and provide examples. OpenAI GPT-6 Astra was also used to generate the first manuscript draft, which the authors reviewed and polished into the final manuscript. The authors verified and revised the AI-assisted work, including the proofs and examples, and take full responsibility for the results.

\appendix

\section{Classical and quantum spectral gap conventions}
\label{sec:prelim}

\paragraph*{Notation.}
For a finite set $S$, $\ell^2(S)$ denotes $\mathbb C^S$ with orthonormal
basis $\{|s\rangle:s\in S\}$. We write $[n]:=\{1,\ldots,n\}$.
For a superoperator $\Phi$, $\Phi^\dagger$ denotes its adjoint with
respect to the Hilbert--Schmidt inner product $\Tr(A^\dagger B)$.
When $\cL$ acts on observables in the Heisenberg picture,
$\cL_*:=\cL^\dagger$ denotes the corresponding generator acting on
states in the Schr\"odinger picture.
These conventions follow \cite[Chs.~1--3]{Watrous2018}.

\subsection{Dirichlet forms and Poincar\'e gaps}
\label{subsec:poincare}

We use the spectral and Dirichlet-form conventions of
\cite[Secs.~12.1 and~13.2]{LevinPeresWilmer2017}.
Let $\Omega$ be a finite set with $|\Omega|\ge2$, and let $R$ be a
transition matrix on $\Omega$. Its support graph has an edge $\{f,g\}$
whenever $f\ne g$ and at least one of $R(f,g)$ and $R(g,f)$ is positive.

Suppose henceforth that $R$ is reversible with respect to a full-support
distribution $\nu$. Write $\nu_f:=\nu(f)$ and define
$\langle a,b\rangle_\nu:=\sum_{f\in\Omega}\nu_f\overline{a_f}b_f$,
$\|a\|_\nu^2:=\langle a,a\rangle_\nu$, and
$L^2(\Omega,\nu):=(\mathbb C^\Omega,\langle\cdot,\cdot\rangle_\nu)$.
Let $\mathbf1$ denote the constant-one function. For
$a\in\mathbb C^\Omega$, define
\begin{align}
 \cE_{\nu,R}(a)
 &:=
 \frac12\sum_{f,g\in\Omega}
 \nu_fR(f,g)|a_f-a_g|^2
 =\langle a,(I-R)a\rangle_\nu,
 \notag\\
 \Var_\nu(a)
 &:=
 \inf_{c\in\mathbb C}
 \sum_{f\in\Omega}\nu_f|a_f-c|^2
 =
 \|a-\langle\mathbf1,a\rangle_\nu\mathbf1\|_\nu^2.
 \label{eq:classical-form}
\end{align}
These are the  Dirichlet form and variance; see
\cite[Lemma~13.6 and Remark~13.8]{LevinPeresWilmer2017}.

Reversibility makes $R$ self-adjoint on $L^2(\Omega,\nu)$
\cite[Lemma~12.2]{LevinPeresWilmer2017}. Its eigenvalues are therefore real
and may be ordered, with multiplicity, as
$1=\lambda_1(R)\ge\lambda_2(R)\ge\cdots$. The Rayleigh--Ritz
characterization \cite[Lemma~13.7 and Remark~13.8]{LevinPeresWilmer2017}
gives
\begin{equation}
 1-\lambda_2(R)
 =
 \inf_{\Var_\nu(a)>0}
 \frac{\cE_{\nu,R}(a)}{\Var_\nu(a)}.
 \label{eq:classical-poincare}
\end{equation}
The quantity $1-\lambda_2(R)$ is the \emph{Poincar\'e spectral gap}. Note that it is
not the absolute discrete-time spectral gap
$1-\max_{j\ge2}|\lambda_j(R)|$.

The multiplicity of the eigenvalue $1$ equals the number of connected
components of the support graph: by \eqref{eq:classical-form}, zero
Dirichlet form means constancy on each connected component
\cite[Secs.~1.3 and~12.1]{LevinPeresWilmer2017}. In particular,
\begin{equation}
 1-\lambda_2(R)>0
 \quad\Longleftrightarrow\quad
 R\text{ is irreducible}.
 \label{eq:gap-irreducible}
\end{equation}

\subsection{Quantum Markov semigroups and the KMS spectral gap}
\label{subsec:qms}

We use the finite-dimensional quantum Markov semigroup conventions of
\cite[Sec.~II]{KastoryanoTemme2013} and the KMS inner product and
symmetry convention of \cite[Secs.~1--2]{FagnolaUmanita2010}.
A quantum Markov semigroup in the Heisenberg picture consists of
completely positive unital maps $T_t=e^{t\cL}$, while the \textit{trace-dual semigroup}, i.e., the semigroup acting on quantum states (the Schrödinger-picture semigroup),
is $T_{t,*}=e^{t\cL_*}$. A state $\sigma$ is stationary if
$\cL_*(\sigma)=0$.
\begin{definition}[Spectral gap of a KMS-symmetric generator~\cite{KastoryanoTemme2013}]
\label{def:gap}
For a full-rank state $\sigma$ on $\cH$ with $\dim\cH\ge2$, set
$\langle X,Y\rangle_{\sigma,1/2}:=\Tr(\sigma^{1/2}X^\dagger
\sigma^{1/2}Y)$, $\Gamma_\sigma(X):=\sigma^{1/2}X\sigma^{1/2}$, and
$\Fix(\cL):=\ker\cL=\{X:e^{t\cL}(X)=X\text{ for every }t\ge0\}$.
Assume that $\cL$ is KMS symmetric with respect to the full-rank
stationary state $\sigma$, as in \Cref{sec:ckg}, and define
\begin{equation}
 \gap(-\cL):=
 \inf_{\substack{X\ne0\\
       \langle I,X\rangle_{\sigma,1/2}=0}}
 \frac{-\langle X,\cL(X)\rangle_{\sigma,1/2}}
      {\langle X,X\rangle_{\sigma,1/2}}.                 \label{eq:gap-definition}
\end{equation}
\end{definition}
Following \cite[Sec.~II]{KastoryanoTemme2013}, we call the semigroup
\emph{primitive} if it has a unique full-rank stationary state.
The following provides direct consequences of KMS symmetry
\cite[Sec.~II, Definitions~6 and~10]{KastoryanoTemme2013}; see also
\cite[Sec.~2]{FagnolaUmanita2010}.

\begin{lemma}[KMS symmetry and the spectral gap~\cite{KastoryanoTemme2013,FagnolaUmanita2010}]
\label{lem:gap-primitive}
\label{lem:KMS-intertwining}
Let $\cL$ generate a quantum Markov semigroup on a finite-dimensional
$\cB(\cH)$ with $\dim\cH\ge2$, and let $\sigma$ be a full-rank state. If $\cL$ is KMS
symmetric with respect to $\sigma$, then it is diagonalizable in the KMS
inner product, its eigenvalues are real and nonpositive, and $\cL(I)=0$.
Moreover,
\[
 \gap(-\cL)>0
 \quad\Longleftrightarrow\quad
 \Fix(\cL)=\mathbb CI,
\]
and when these hold, $\gap(-\cL)$ is the smallest nonzero eigenvalue of
$-\cL$ and
\begin{equation}
 -\langle X,\cL(X)\rangle_{\sigma,1/2}
 \ \ge\ \gap(-\cL)\inf_{c\in\mathbb C}\|X-cI\|_{\sigma,1/2}^2
 \qquad\text{for every }X\in\cB(\cH),                   \label{eq:KMS-poincare}
\end{equation}
where $\|Y\|_{\sigma,1/2}^2:=\langle Y,Y\rangle_{\sigma,1/2}$.
For any finite-dimensional quantum Markov generator and full-rank state
$\sigma$, KMS symmetry is equivalent to
\begin{equation}
 \cL_*\circ\Gamma_\sigma=\Gamma_\sigma\circ\cL .
 \label{eq:KMS-intertwining-general}
\end{equation}
In that case $e^{t\cL_*}\Gamma_\sigma=\Gamma_\sigma e^{t\cL}$ for all
$t\ge0$, $\cL_*(\sigma)=0$, and the stationary space is
$\mathbb C\sigma$ when $\Fix(\cL)=\mathbb CI$.
\end{lemma}

\begin{proof}
KMS symmetry is self-adjointness in $\langle\cdot,\cdot\rangle_{\sigma,1/2}$,
so the finite-dimensional spectral theorem gives a real spectrum and an
orthonormal eigenbasis. Each $e^{t\cL}$ is unital and completely positive,
hence contractive in operator norm; consequently no eigenvalue of $\cL$
is positive. Since $\cL(I)=0$, restriction to the KMS-orthogonal complement
of $I$ gives the gap characterization and \eqref{eq:KMS-poincare}.
The identity
$\langle X,Y\rangle_{\sigma,1/2}=\Tr(X^\dagger\Gamma_\sigma(Y))$
identifies KMS symmetry with \eqref{eq:KMS-intertwining-general}.
Exponentiating and applying to $I$ gives the semigroup relation and
$\cL_*(\sigma)=0$. Finally, invertibility of $\Gamma_\sigma$ gives
$\ker\cL_*=\Gamma_\sigma(\ker\cL)$, proving the assertion about the
stationary space.
\end{proof}

\section{Quantum channels associated with Markov chains}
\label{sec:associated-lindbladian}

We use the finite-dimensional operator-sum and dual-map
conventions of \cite[Sec.~2.2]{Watrous2018}; the stochastic-matrix channel
below is the Heisenberg-picture dual of the quantum--classical channel in
\cite[Sec.~3.2]{AhiableEtAl2021}.

\begin{definition}[Measure-and-prepare channel associated with a Markov chain]
\label{def:Markov-kernel-embedding}
Let $F$ be finite with $N:=|F|\ge2$, and let $Q$ be the transition matrix
of an irreducible Markov chain on $F$ satisfying $Q=Q^{\mathsf T}$.  Let
$\cH_F:=\ell^2(F)$ with
canonical basis $\{|f\rangle:f\in F\}$.  Define
\begin{equation}
 K_{fg}:=\sqrt{Q(f,g)}\,|g\rangle\langle f|,
 \qquad
 \Phi_Q(A):=\sum_{f,g\in F}K_{fg}^\dagger A K_{fg},
 \qquad
 \cL_0:=\Phi_Q-\id.                                      \label{eq:lift}
\end{equation}
Here $\cL_0$ is the \textit{normalized} infinite-temperature generator or Lindbladian associated
with the coupling family $\{K_{fg}\}$. Row-stochasticity gives $\sum_{f,g}K_{fg}^\dagger K_{fg}=I$, and
$K_{gf}=K_{fg}^\dagger$ because $Q=Q^{\mathsf T}$. Thus the family is
closed under adjoints, with involution $\iota(f,g)=(g,f)$.
Let $u_f:=1/N$ be the uniform distribution on $F$ and $\tau:=I/N$.  The associated Dirichlet form is
$\cE_0(A):=-\langle A,\cL_0(A)\rangle_\tau$, with the inner product and
variance of \Cref{ass:poincare}.
For $A\in\cB(\cH_F)$ write $a:=(A_{ff})_{f\in F}$ for its diagonal and
$\cD(A):=\sum_fA_{ff}|f\rangle\langle f|$ for its diagonal part.
\end{definition}

The following lemma gives the block decomposition and spectral gap of
$\cL_0$; compare the channel construction and
spectral correspondence in \cite[Sec.~3.2, Theorem~3.1 and Corollary~3.2]{AhiableEtAl2021}.

\begin{lemma}[Block structure and spectral gap of the normalized infinite-temperature generator]
\label{lem:associated-structure}
\label{thm:associated-gap}
Adopt \Cref{def:Markov-kernel-embedding}. Then $\cL_0$ is a GKSL generator
that is KMS symmetric with respect to $\tau$, and, for every
$A\in\cB(\cH_F)$,
\begin{equation}
 -\cL_0(A)=
 \sum_{f\in F}\bigl((I-Q)a\bigr)_f|f\rangle\langle f|
 +A-\cD(A).                                               \label{eq:L0-block}
\end{equation}
Moreover, $\Fix(\cL_0)=\mathbb CI$ and
$\ker\cL_{0,*}=\mathbb C\tau$. Consequently, $e^{t\cL_0}$ is primitive,
$\tau$ is its unique stationary state, and
\[
 \gap(-\cL_0)
 =\lambda_Q
 :=\min\{1,1-\lambda_2(Q)\}
 \in(0,1].
\]
\end{lemma}

\begin{proof}
Direct multiplication gives
\begin{equation}
 \Phi_Q(A)=\sum_{f\in F}(Qa)_f\,|f\rangle\langle f| .
 \label{eq:PhiQ-diagonal}
\end{equation}
Since $\sum_{f,g}K_{fg}^\dagger K_{fg}=I$,
$\cL_0(A)=\sum_{f,g}K_{fg}^\dagger AK_{fg}-\frac12\{I,A\}$
is in GKSL form. The diagonal and off-diagonal subspaces are orthogonal
and invariant: $-\cL_0$ acts as $I-Q$ on the former and as the identity
on the latter. Symmetry of $Q$ gives KMS symmetry with respect to $\tau$,
and irreducibility makes the zero eigenvalue simple. The eigenvalues of
$-\cL_0$ are therefore $1-\lambda_i(Q)$ for $1\le i\le N$, together with
$N(N-1)$ additional copies of $1$, proving the assertions.
The gap is attained by a diagonal eigenoperator for $\lambda_2(Q)$ if
$1-\lambda_2(Q)\le1$, and by any off-diagonal matrix unit otherwise.
\end{proof}

\subsection{Kraus representations of \texorpdfstring{$\Phi_Q$}{Phi Q}}
\label{subsec:kraus-families}

Kraus representations of the same completely positive map are related
by unitary mixing after padding with zeros
\cite[Corollary~2.23]{Watrous2018}.

\begin{definition}[Kraus family]
\label{def:kraus-family}
Let $\mathcal A$ be a finite index set. A family
$\{K_a\}_{a\in\mathcal A}\subseteq\mathcal B(\mathcal H_F)$ is a
\emph{Kraus family for $\Phi_Q$} if
$\sum_{a\in\mathcal A}K_a^\dagger K_a=I$ and
$\Phi_Q(A)=\sum_{a\in\mathcal A}K_a^\dagger A K_a$ for every
$A\in\mathcal B(\mathcal H_F)$.
Adjoint closure up to phases is as in \Cref{ass:ckg}; the family is
\emph{closed under adjoints} if the phases may all be chosen equal to $1$.
\end{definition}

Adjoint closure and normalization also give $\sum_aK_aK_a^\dagger=I$.

\begin{lemma}[Commutator representation of the infinite-temperature Dirichlet form]
\label{lem:commutator-sum}
\label{lem:carre-L0}
Let $\cH$ be a Hilbert space of finite dimension $N\ge2$, put $\tau:=I/N$,
and let $\{L_a\}_{a\in\mathcal A}$ be a finite family in $\cB(\cH)$ with
\begin{equation}
 \sum_{a\in\mathcal A}L_a^\dagger L_a=I=\sum_{a\in\mathcal A}L_aL_a^\dagger .
 \label{eq:double-completeness}
\end{equation}
Put $\Psi(O):=\sum_aL_a^\dagger OL_a$. Then $\Psi$ is unital, trace
preserving, completely positive, and a contraction for
$\|\cdot\|_\tau$, and for every $O\in\cB(\cH)$,
\begin{equation}
 \frac12\sum_{a}\|[L_a,O]\|_\tau^2
 =\|O\|_\tau^2-\Re\langle O,\Psi(O)\rangle_\tau
 \ \le\ 2\|O\|_\tau^2 .
 \label{eq:commutator-sum}
\end{equation}
In particular, for any Kraus family $\{K_a\}$ for $\Phi_Q$ that is closed
under adjoints up to phases as in \Cref{def:kraus-family},
\[
 \frac12\sum_a\|[K_a,O]\|_\tau^2
 =\|O\|_\tau^2-\Re\langle O,\Phi_Q(O)\rangle_\tau
 =\cE_0(O)
 \ge\lambda_Q\Var_\tau(O).
\]
\end{lemma}

\begin{proof}
The completeness relations give unitality and trace preservation.
Expanding the commutators and using cyclicity of the trace proves
\eqref{eq:commutator-sum}. The Schwarz inequality for unital completely
positive maps \cite[Corollary~2.8]{Choi1974} gives
$\|\Psi(O)\|_\tau\le\|O\|_\tau$; Cauchy--Schwarz proves the bound.
For $\Phi_Q$, adjoint closure gives \eqref{eq:double-completeness},
self-adjointness makes $\langle O,\Phi_Q(O)\rangle_\tau$ real, and
\Cref{thm:associated-gap} gives the Poincar\'e inequality.
\end{proof}

\paragraph{Hermitian Kraus representation.}
Let $\{K_a\}_{a\in\mathcal A}$ be a Kraus family for $\Phi_Q$ that is
closed under adjoints, with involution $\iota$.  Let $\Lab_1$ be the set of fixed points of
$\iota$ and $\Lab_2$ a set of representatives of its two-cycles.
Define the new index set
$\Lab_H:=\Lab_1\sqcup(\Lab_2\times\{+,-\})$.  Set $A_a:=K_a$ for
$a\in\Lab_1$ and
\begin{equation}
 A_{a,+}:=\frac{K_a+K_a^\dagger}{\sqrt2},
 \qquad
 A_{a,-}:=\frac{i(K_a-K_a^\dagger)}{\sqrt2},
 \qquad a\in\Lab_2.                                \label{eq:Hermitian-jumps}
\end{equation}
For each two-cycle, this is the unitary recombination
$2^{-1/2}\bigl(\begin{smallmatrix}1&1\\ i&-i\end{smallmatrix}\bigr)$.
By \cite[Corollary~2.23]{Watrous2018}, $\{A_b\}_{b\in\Lab_H}$ is a
Kraus family for $\Phi_Q$ consisting of Hermitian operators; in particular
$\sum_bA_b^2=I$ and $\sum_bA_bOA_b=\Phi_Q(O)$.

\subsection{Combining infinite-temperature Poincar\'e inequalities}
\label{subsec:mixture}

On the same finite-dimensional Hilbert space $\cH$, let each finite family
$\{A_{j,a}\}_a$ be normalized and closed under adjoints up to phases, and
suppose its Poincar\'e inequality at infinite temperature holds with constant $\kappa_j\ge0$.
For finitely many weights $w_j\ge0$ with $\sum_jw_j=1$, the family
$\{\sqrt{w_j}A_{j,a}\}_{j,a}$ is normalized and closed under adjoints up to
phases. Its commutator Dirichlet form is the weighted sum of the individual
forms, so their Poincar\'e inequalities add:
\[
 \frac12\sum_{j,a}\|[\sqrt{w_j}A_{j,a},O]\|_\tau^2
 \ge\Bigl(\sum_jw_j\kappa_j\Bigr)\Var_\tau(O).
\]
In particular, including the reset family
$\{D^{-1/2}\ket f\bra g\}_{f,g}$ with weight $\eta\in(0,1]$, where
$\{\ket f\}_{f=1}^D$ is an orthonormal basis of $\cH$, gives
infinite-temperature Poincar\'e constant at least $\eta$: its normalized
generator $X\mapsto D^{-1}\Tr(X)I-X$ has spectral gap one.
This does not by itself establish efficient access to the resulting couplings.

\section{CKG construction and invariances}
\label{app:ckg-details}

This appendix fixes the conventions for the CKG generator used in
\Cref{sec:ckg} and relates it to the normalized infinite-temperature
generator. We follow the construction and KMS detailed-balance result of
\cite{CKGTechnical,DLL2025} under \Cref{ass:ckg}; the additional
normalization in \Cref{ass:poincare} enters when identifying the
infinite-temperature limit.

\subsection{Construction and detailed balance}

At inverse temperature $\beta>0$, the construction begins by filtering
the Heisenberg evolution of each coupling with the normalized Gaussian
\begin{align}
 f_\beta(t)&:=\frac{e^{-t^2/\beta^2}}
 {\sqrt{\beta\sqrt{\pi/2}}},
 &
 \widehat f_\beta(\omega)&:=\frac1{\sqrt{2\pi}}\int_{\mathbb R}
 f_\beta(t)e^{-i\omega t}\,dt=
 \frac{\sqrt\beta}{(2\pi)^{1/4}}e^{-\beta^2\omega^2/4},  \label{eq:filter}
\end{align}
in the convention of \cite[Eqs.~(7)--(8)]{CKGTechnical}.
Here $\int f_\beta^2=\int\widehat f_\beta^{\,2}=1$ by Plancherel
\cite[Theorem~9.13]{Rudin1987}, and
$\int t^2f_\beta(t)^2\,dt=\beta^2/4$.
The resulting energy-filtered coupling operators are
\begin{equation}
 \widehat A_a(\omega):=\frac1{\sqrt{2\pi}}
 \int_{\mathbb R}f_\beta(t)e^{-i\omega t}
 e^{iHt}A_ae^{-iHt}\,dt.                                  \label{eq:OFT}
\end{equation}
These operators enter the dissipative part of the generator with either
of the CKG transition weights
\begin{align}
 \gamma_{\rm G}(\omega):=e^{-(\beta\omega+1)^2/2},\quad
 \gamma_{\rm M}(\omega):=
 e^{-\beta\max\{\omega+1/(2\beta),0\}}.                \label{eq:weights}
\end{align}
These are the weights in \cite[Eqs.~(10) and~(14)]{CKGTechnical}. Both are
measurable functions $\mathbb R\to(0,1]$; the bound $\gamma\le1$ will imply
$R_\gamma\preceq I$.

To specify the coherent correction accompanying these dissipative terms,
we resolve operators into Bohr frequencies. Let
$H=\sum_{E\in\spec(H)}E\,P_E$ and
$\Omega(H):=\{E'-E:E,E'\in\spec(H)\}$. Every $X\in\cB(\cH)$ then has
the decomposition
\begin{equation}
 X=\sum_{\nu\in\Omega(H)}X_\nu,
 \qquad
 X_\nu:=\sum_{\substack{E,E'\in\spec(H)\\E'-E=\nu}}P_{E'}XP_E,
 \label{eq:bohr-components}
\end{equation}
as in \cite[Sec.~II.1, Eqs.~(32)--(33)]{CKGTechnical}.
With $\ad_H(X):=[H,X]$ and
$\sigma_\beta=e^{-\beta H}/\Tr(e^{-\beta H})$, this decomposition gives
\begin{equation}
 \begin{aligned}
 &\ad_H(X_\nu)=\nu X_\nu,\quad
 (X_\nu)^\dagger=(X^\dagger)_{-\nu},\\
 &\sigma_\beta^{1/2}X_\nu\sigma_\beta^{-1/2}=e^{-\beta\nu/2}X_\nu,\quad
 \widehat A_a(\omega)=\sum_\nu\widehat f_\beta(\omega-\nu)(A_a)_\nu,
 \end{aligned}
 \label{eq:bohr}
\end{equation}
where all sums over $\nu,\mu$ are over $\Omega(H)$; compare
\cite[Sec.~II.1, Eqs.~(32)--(34)]{CKGTechnical}.
Thus the same components describe both the frequency filtering and the
conjugation by $\sigma_\beta^{1/2}$. They also define the functional calculus
needed for the coherent correction: for $h:\Omega(H)\to\mathbb C$, set
\begin{equation}
 h(\ad_H)(X):=\sum_{\nu\in\Omega(H)}h(\nu)X_\nu.
 \label{eq:ad-calculus}
\end{equation}
Define the auxiliary space and operators by
\begin{align}
 \cK_\gamma&:=\ell^2(\mathcal A)\otimes L^2(\mathbb R,d\omega),\notag\\
 (V_\gamma|\psi\rangle)_{a,\omega}
 &:=\sqrt{\gamma(\omega)}\,\widehat A_a(\omega)|\psi\rangle,
 \qquad R_\gamma:=V_\gamma^\dagger V_\gamma
 =\sum_a\int_{\mathbb R}\gamma(\omega)\widehat A_a(\omega)^\dagger
 \widehat A_a(\omega)\,d\omega,\notag\\
 B_\gamma&:=\frac{i}{2}\tanh\!\left(
 \frac\beta4\ad_H\right)(R_\gamma),
 \qquad \ad_H(X):=[H,X].                   \label{eq:CKG-data}
\end{align}
Here $V_\gamma:\cH\to\cK_\gamma\otimes\cH$ collects the filtered jump
operators, and $B_\gamma$ is the correction prescribed by
\cite[Lemma~II.1, Eq.~(45)]{CKGTechnical}.
The resulting exact CKG generator in the Schr\"odinger picture is
\begin{equation}
 \cL_{\gamma,*}(\rho)=
 -i[B_\gamma,\rho]
 +\Tr_{\cK_\gamma}(V_\gamma\rho V_\gamma^\dagger)
 -\frac12\{R_\gamma,\rho\}
 =-i[B_\gamma,\rho]+\cT_\gamma(\rho)-\frac12\{R_\gamma,\rho\},
                                                       \label{eq:CKG-lindbladian}
\end{equation}
with gain map
$\cT_\gamma(\rho):=\sum_a\int\gamma(\omega)\widehat A_a(\omega)\rho
\widehat A_a(\omega)^\dagger\,d\omega$. We write
$\cL_\gamma:=\cL_{\gamma,*}^{\dagger}$ for the Heisenberg-picture
generator.

The coupling normalization in \Cref{ass:ckg} ensures that these
expressions are well defined and that
\begin{equation}
 \|V_\gamma\|\le1,\qquad
 0\preceq R_\gamma\preceq I,\qquad
 B_\gamma^\dagger=B_\gamma.
 \label{eq:ckg-data-bounded}
\end{equation}
To see this, put $A_a(t):=e^{iHt}A_ae^{-iHt}$. Plancherel,
$\gamma\le1$, and $\sum_aA_a(t)^\dagger A_a(t)\preceq I$ give
\[
 \|V_\gamma\psi\|^2
 \le\sum_a\int\|\widehat A_a(\omega)\psi\|^2d\omega
 =\int f_\beta(t)^2\sum_a\|A_a(t)\psi\|^2dt
 \le\|\psi\|^2.
\]
Hence $V_\gamma$ is a contraction and
$0\preceq R_\gamma=V_\gamma^\dagger V_\gamma\preceq I$.
Summing this bound over an orthonormal basis gives absolute norm
convergence of the integrals defining $R_\gamma$ and $\cT_\gamma$.
The Bohr identities in \eqref{eq:bohr} give
$(R_\gamma)_\nu^\dagger=(R_\gamma)_{-\nu}$, so oddness of $\tanh$
makes $B_\gamma$ Hermitian. Meanwhile, the representation
$\cT_\gamma(\rho)=\Tr_{\cK_\gamma}(V_\gamma\rho V_\gamma^\dagger)$
gives complete positivity. The operators
$\sqrt{\gamma(\omega)}\widehat A_a(\omega)$ form an integral Kraus
representation with loss operator $R_\gamma$; in finite dimension,
$\cT_\gamma$ also has a finite Kraus representation
\cite[Theorem~2.22]{Watrous2018}. Substituting it into
\eqref{eq:CKG-lindbladian} therefore gives a bounded GKSL generator
with Hamiltonian part $B_\gamma$
\cite{GoriniKossakowskiSudarshan1976,Lindblad1976}.

To apply the CKG detailed-balance theorem, we take the couplings to be
closed under adjoints. Although \Cref{ass:ckg} allows closure up to
phases, the following invariance shows that those phases can be removed
without changing the generator. The same observation justifies using
any Kraus representation of the coupling map in the applications.

\begin{lemma}[Invariance under isometric mixing of Kraus operators]
\label{lem:relabelling}
Let $\{A_a\}_{a\in\mathcal A}$ and $\{A'_b\}_{b\in\mathcal B}$ be finite families with
$A'_b=\sum_{a}U_{ba}A_a$, where
$U:\ell^2(\mathcal A)\to\ell^2(\mathcal B)$ is an isometry, so
$U^\dagger U=I_{\ell^2(\mathcal A)}$.
Then the two families define the same $R_\gamma$, $B_\gamma$,
$\cT_\gamma$, and hence the same generator $\cL_{\gamma,*}$, for every
$\beta>0$ and every bounded nonnegative transition weight $\gamma$.  In particular:
\begin{enumerate}
\item[(i)] multiplying the coupling operators by label-dependent phases does not
change $\cL_{\gamma,*}$;
\item[(ii)] a family that is closed under adjoints up to phases defines the same
generator as a family closed under adjoints: the phases can be removed
by choosing
$\theta_a+\theta_{\iota(a)}=-\varphi_a$ modulo $2\pi$ and
replacing $A_a$ by $e^{i\theta_a}A_a$;
\item[(iii)] any two Kraus representations of the same completely positive
map define the same exact CKG generator, at every temperature.
\end{enumerate}
\end{lemma}

\begin{proof}
Linearity gives $\widehat A'_b(\omega)=\sum_aU_{ba}\widehat A_a(\omega)$.
Since $U^\dagger U=I$, the sums defining $\cT_\gamma$ and $R_\gamma$
are unchanged at each frequency, hence so are $B_\gamma$ and the generator.
This proves (i). For (ii), adjoint-closure gives
$\varphi_{\iota(a)}\equiv\varphi_a\pmod{2\pi}$ for nonzero operators.
Choose $\theta_a=-\varphi_a/2$ on fixed points and
$(\theta_a,\theta_{\iota(a)})=(-\varphi_a,0)$ on each two-cycle;
zero operators impose no constraint. Then the rephased family is closed
under adjoints. For (iii), use the unitary equivalence of Kraus
representations after padding with zeros
\cite[Corollary~2.23]{Watrous2018}.
\end{proof}

After this phase removal, \cite[Theorem~I.1]{CKGTechnical} applies to
either transition weight in \eqref{eq:weights}. To match its conventions,
set $\sigma_E=1/\beta$ and, for the Gaussian weight,
$\sigma_\gamma=\omega_\gamma=1/\beta$. These choices give
\eqref{eq:filter}--\eqref{eq:CKG-lindbladian}, with
\cite[Lemma~II.1]{CKGTechnical} prescribing the coherent correction.
Writing $\sigma:=\sigma_\beta$ and using $\Gamma_\sigma$ from \Cref{def:gap},
the theorem gives
\begin{equation}
 \cL_{\gamma,*}\circ\Gamma_\sigma=\Gamma_\sigma\circ\cL_\gamma.
 \label{eq:KMS-DB}
\end{equation}
Their Schr\"odinger-picture detailed-balance convention
\cite[Definition~II.1]{CKGTechnical} is
$\cL_\gamma=\Gamma_\sigma^{-1}\circ\cL_{\gamma,*}\circ\Gamma_\sigma$,
which is precisely \eqref{eq:KMS-DB}. Thus $\cL_\gamma$ is KMS
symmetric with respect to $\langle\cdot,\cdot\rangle_{\sigma,1/2}$.
Applying \eqref{eq:KMS-DB} to $I$ also gives
$\cL_{\gamma,*}(\sigma)=0$, since $\Gamma_\sigma(I)=\sigma$ and
$\cL_\gamma(I)=0$. Thus the Gibbs state is stationary for the corrected CKG generator,
as required in the spectral gap comparison.

\subsection{Temperature scaling and normalization}
\label{app:infinite temperature-normalization}

The same formulas determine how the generator depends on the energy
scale and inverse temperature. For fixed couplings and either choice of
transition weight, replacing $(H,\beta)$ by $(cH,\beta/c)$ with $c>0$
leaves $R_\gamma$, $B_\gamma$, and $\cT_\gamma$ unchanged. With
superscripts indicating these parameters, this gives
\begin{equation}
 \cL_{\gamma,*}^{(cH,\beta/c)}=\cL_{\gamma,*}^{(H,\beta)}.
 \label{eq:ckg-scaling}
\end{equation}
Indeed, for $\beta'=\beta/c$ and $H'=cH$, changing variables $s=ct$ gives
$f_{\beta'}(t)=\sqrt c\,f_\beta(ct)$,
$\widehat A'_a(\omega')=c^{-1/2}\widehat A_a(\omega'/c)$, and
$\gamma'(\omega')=\gamma(\omega'/c)$. Thus
$V'_\gamma=(I\otimes D_c\otimes I)V_\gamma$, where the unitary
frequency rescaling is $(D_c\phi)(\omega')=c^{-1/2}\phi(\omega'/c)$.
Substitution $\omega'=c\omega$ leaves $R_\gamma$ and $\cT_\gamma$
unchanged; $\beta'\ad_{H'}=\beta\ad_H$ gives the same $B_\gamma$
and generator. The generator therefore depends on $(H,\beta)$
only through $\beta H$.

To compute the infinite-temperature limit, define
$q_\gamma:=\int_{\mathbb R}\gamma(\omega)\widehat f_\beta(\omega)^2d\omega$.
After the substitution $x=\beta\omega$ and completion of the square,
$q_{\gamma_{\rm G}}=e^{-1/4}/\sqrt2\approx0.551$ and
$q_{\gamma_{\rm M}}=\operatorname{erfc}(1/(2\sqrt2))\approx0.617$, where
$\operatorname{erfc}(x):=2\pi^{-1/2}\int_x^\infty e^{-t^2}dt$ is the
complementary error function. Both constants lie in $(0,1]$ and are
independent of $\beta$.

These constants relate the CKG generator to the normalized
infinite-temperature generator used in the main comparison. For fixed
$H$ and couplings satisfying \Cref{ass:poincare}, put
$\Phi_{\mathcal A}(X):=\sum_aA_a^\dagger X A_a$ and write
$\cL_{\beta,\gamma}$ for the Heisenberg-picture CKG generator with either
transition weight. Then
\begin{equation}
 \lim_{\beta\downarrow0}\cL_{\beta,\gamma}
 =q_\gamma\cL_{\mathcal A,0}.
 \label{eq:infinite temperature-limit}
\end{equation}
The limit holds in every norm on this finite-dimensional space of
superoperators. To verify it, use \eqref{eq:bohr} and put
$s=\beta\omega$. The transition weight becomes a fixed function of
$s$, and each coefficient
$\int\gamma(\omega)\widehat f_\beta(\omega-\nu)
\widehat f_\beta(\omega-\mu)\,d\omega$ tends to $q_\gamma$.
Dominated convergence applies because there are finitely many
frequencies and the translated Gaussian products have a common
integrable bound for sufficiently small $\beta$.
Thus the Heisenberg-picture gain map tends to
$q_\gamma\Phi_{\mathcal A}$ and $R_\gamma\to q_\gamma I$, while
\eqref{eq:CKG-data} gives $B_\gamma\to0$.
Throughout the paper, spectral gap bounds and Poincar\'e constants at
infinite temperature refer to the normalized generator $\cL_{\mathcal A,0}$,
unless the factor $q_\gamma$ is explicitly included.

\section{Proofs of the spectral gap comparisons}
\label{app:gap-comparison-proofs}

This appendix proves the comparison results supporting
\Cref{sec:comparison-main}. The stationary commutator identity is the
common starting point. We first prove the spectral gap bound valid at every finite inverse
temperature, then state and prove the complementary perturbative comparison.
The final corollary applies these bounds to couplings obtained from a
symmetric Markov chain. Trace-distance mixing bounds are given next in
\Cref{sec:mixing}.

\subsection{The stationary commutator identity}

Use the Gelfand--Naimark--Segal (GNS) inner product
$\langle X,Y\rangle_{\sigma,\mathrm{GNS}}:=\Tr(\sigma X^\dagger Y)$.
The associated squared norm is
$\|X\|_{\sigma,\mathrm{GNS}}^2:=\langle X,X\rangle_{\sigma,\mathrm{GNS}}$,
and its variance is
$\Var_{\sigma,\mathrm{GNS}}(X):=\inf_c\|X-cI\|_{\sigma,\mathrm{GNS}}^2
=\|X\|_{\sigma,\mathrm{GNS}}^2-|\Tr(\sigma X)|^2$.
A generator is GNS symmetric with respect to $\sigma$ if it is self-adjoint for $\langle\cdot,\cdot\rangle_{\sigma,\mathrm{GNS}}$. This is the $s=1$ GNS convention of \cite[Eq.~(2.8) and Definition~4]{DLL2025}.

\begin{lemma}[Stationary commutator identity~\cite{Lindblad1976}]
\label{lem:carre}
Adopt \Cref{ass:ckg}, fix $\beta>0$ and
$\gamma\in\{\gamma_{\rm G},\gamma_{\rm M}\}$, and put
$\sigma:=e^{-\beta H}/\Tr(e^{-\beta H})$.  For every $X\in\cB(\cH)$,
\[
 -\Re\Tr(\sigma X^\dagger\cL_\gamma(X))
 =\frac12\sum_a\int_{\mathbb R}\gamma(\omega)
 \|[\widehat A_a(\omega),X]\|_{\sigma,\mathrm{GNS}}^2\,d\omega.
\]
The Hamiltonian part cancels in
$\cL_\gamma(X^\dagger X)-\cL_\gamma(X^\dagger)X-X^\dagger\cL_\gamma(X)$.
Taking its expectation in the stationary state $\sigma$ gives the identity.
\end{lemma}

\begin{proof}
The GKSL dissipation identity \cite[Sec.~3]{Lindblad1976} is
\[
 \cL(X^\dagger X)-\cL(X^\dagger)X-X^\dagger\cL(X)
 =\sum_j[L_j,X]^\dagger[L_j,X].
\]
Apply it to the jump operators $\sqrt{\gamma(\omega)}\widehat A_a(\omega)$
and integrate, using the norm convergence established following \eqref{eq:ckg-data-bounded}. Taking the expectation in the stationary state
$\sigma$ eliminates the first term; the other two terms are complex
conjugates. This gives the claimed identity without assuming GNS symmetry.
\end{proof}

\subsection{Proof of the all-temperature comparison}
\label{app:comparison-proof}

\begin{proof}[Proof of \Cref{thm:metropolis-all-temperature}]
Write $\sigma:=\sigma_\beta$, and let $p_{\min}$ and $p_{\max}$ be its
smallest and largest eigenvalues.  Then
$p_{\min}/p_{\max}=e^{-\beta W}=e^{-b}$.  First take $X=X^\dagger$, set
$Y_{a,\omega}:=[\widehat A_a(\omega),X]$, and define
\begin{equation}
 I_{\rm M}(X):=\sum_a\int_{\mathbb R}\gamma_{\rm M}(\omega)
 \Tr(\sigma Y_{a,\omega}^\dagger Y_{a,\omega})\,d\omega .
 \label{eq:IM-definition}
\end{equation}
By \Cref{lem:carre},
\begin{equation}
 -\Re\Tr\!\left(\sigma X\cL_{\gamma_{\rm M}}(X)\right)
 =\frac12I_{\rm M}(X).                                  \label{eq:IM-carre}
\end{equation}

The filter $f_\beta$ is real and even. Combining this with adjoint closure
up to phases in
\eqref{eq:generic-CKG-assumptions} gives
\[
 \widehat A_{\iota(a)}(-\omega)
 =e^{i\varphi_a}\widehat A_a(\omega)^\dagger,
 \qquad
 Y_{\iota(a),-\omega}
 =-e^{i\varphi_a}Y_{a,\omega}^\dagger .
\]
Relabelling $a$ by $\iota(a)$ and changing $\omega$ to $-\omega$ in
\eqref{eq:IM-definition} therefore give the equivalent expression
\begin{equation}
 I_{\rm M}(X)=\sum_a\int_{\mathbb R}\gamma_{\rm M}(-\omega)
 \Tr(\sigma Y_{a,\omega}Y_{a,\omega}^\dagger)\,d\omega .
 \label{eq:IM-adjoint-representation}
\end{equation}
Average \eqref{eq:IM-definition} and
\eqref{eq:IM-adjoint-representation}.  Since $\sigma\succeq p_{\min}I$,
both traces in the resulting integrand are at least
$p_{\min}\|Y_{a,\omega}\|_2^2$.  Moreover, on writing
$x:=\beta\omega$, the definition of $\gamma_{\rm M}$ gives
$\gamma_{\rm M}(\omega)+\gamma_{\rm M}(-\omega)\ge1$.
Indeed, if $|x|\le\frac12$, the left side is
$2e^{-1/2}\cosh x>1$; if $x\ge\frac12$ or $x\le-\frac12$, one of its
two terms equals $1$.  Consequently,
\begin{equation}
 I_{\rm M}(X)
 \ge\frac{p_{\min}}2\sum_a\int_{\mathbb R}
 \|Y_{a,\omega}\|_2^2\,d\omega
 =\frac{Dp_{\min}}2\sum_a\int_{\mathbb R}
 \|Y_{a,\omega}\|_\tau^2\,d\omega .                   \label{eq:IM-HS-lower}
\end{equation}

The function $\omega\mapsto Y_{a,\omega}$ is the Fourier transform of
$t\mapsto f_\beta(t)[A_a(t),X]$, where
$A_a(t):=e^{iHt}A_ae^{-iHt}$.  Plancherel's theorem and unitary invariance
of the Poincar\'e inequality at infinite temperature give
\begin{align}
 \sum_a\int_{\mathbb R}\|Y_{a,\omega}\|_\tau^2\,d\omega
 &=\int_{\mathbb R}f_\beta(t)^2
   \sum_a\|[A_a(t),X]\|_\tau^2\,dt\notag\\
 &\ge2\kappa_0\Var_\tau(X),                            \label{eq:unweighted-Plancherel-lower}
\end{align}
because $\int_{\mathbb R}f_\beta(t)^2dt=1$.  Combining
\eqref{eq:IM-carre}, \eqref{eq:IM-HS-lower}, and
\eqref{eq:unweighted-Plancherel-lower}, and using
$\Var_{\sigma,\mathrm{GNS}}(X)\le Dp_{\max}\Var_\tau(X)$, gives
\begin{equation}
 -\Re\Tr\!\left(\sigma X\cL_{\gamma_{\rm M}}(X)\right)
 \ge\frac12e^{-b}\kappa_0\Var_{\sigma,\mathrm{GNS}}(X)
 \qquad(X=X^\dagger).                                  \label{eq:metropolis-Hermitian-Poincare}
\end{equation}

\Cref{eq:metropolis-Hermitian-Poincare} implies that every Hermitian
fixed operator is scalar.  A Heisenberg-picture GKSL generator preserves
adjoints, so the Hermitian and skew-Hermitian parts of any fixed operator
are also fixed.  Thus $\Fix(\cL_{\gamma_{\rm M}})=\mathbb CI$, and
the detailed-balance identity \eqref{eq:KMS-DB} and
\Cref{lem:gap-primitive} imply primitivity and uniqueness of the stationary state.

Finally, KMS symmetry gives real eigenvalues, and preservation of
adjoints ensures that every eigenspace corresponding to a nonzero eigenvalue
contains a nonzero Hermitian operator.  If
$\cL_{\gamma_{\rm M}}(X)=-\lambda X$ for such an $X$ with
$\lambda\ne0$, then stationarity gives $\Tr(\sigma X)=0$.  Hence
$\Var_{\sigma,\mathrm{GNS}}(X)=\Tr(\sigma X^2)$, and
\eqref{eq:metropolis-Hermitian-Poincare} gives
$\lambda\ge\frac12e^{-b}\kappa_0$.  Taking the smallest nonzero
eigenvalue of $-\cL_{\gamma_{\rm M}}$ proves \eqref{eq:metropolis-all-temperature-gap}.
\end{proof}

\subsection{Perturbative comparison near infinite temperature}
\label{app:perturbative-proof}

The next estimate also covers the Gaussian transition weight. Its proof
bounds the difference between $[\widehat A_a(\omega),X]$ and
$\widehat f_\beta(\omega)[A_a,X]$, and compares the GNS and normalized
Hilbert--Schmidt norms. It does not assume separated Bohr frequencies.

\begin{theorem}[Perturbative comparison near infinite temperature]
\label{thm:comparison}
Adopt \Cref{ass:poincare}, and fix $\beta>0$ and
$\gamma\in\{\gamma_{\rm G},\gamma_{\rm M}\}$.  Set $b:=\beta W$ and
$\delta_\gamma:=\sqrt2e^bb/\sqrt{q_\gamma\kappa_0}$.
If $\delta_\gamma<1$, then the exact CKG
generator is primitive, $\sigma_\beta$ is its unique stationary state, and
\begin{equation} 
 \gap(-\cL_\gamma)\ge
 q_\gamma(1-\delta_\gamma)^2e^{-b}\kappa_0.              \label{eq:comparison-gap}
\end{equation}
No assumption is imposed on the number, spacing, or multiplicity of Bohr
frequencies.  Under the condition $\delta_\gamma<1$, the right side of
\eqref{eq:comparison-gap} is increasing in $\kappa_0$ and decreasing in
$b$.  More precisely, if $0<\underline\kappa\le\kappa_0$,
$\bar b\ge b$, and
$\bar\delta_\gamma:=\sqrt2e^{\bar b}\bar b/
\sqrt{q_\gamma\underline\kappa}<1$,
then
\begin{equation}
 \gap(-\cL_\gamma)\ge
 q_\gamma(1-\bar\delta_\gamma)^2e^{-\bar b}\underline\kappa.
 \label{eq:comparison-substituted-gap}
\end{equation}
\end{theorem}

\begin{proof}
Let $Z_\beta:=\Tr(e^{-\beta H})$.  The smallest and largest eigenvalues of
$\sigma=\sigma_\beta$ are respectively
$p_{\min}=e^{-\beta\lambda_{\max}(H)}/Z_\beta$ and
$p_{\max}=e^{-\beta\lambda_{\min}(H)}/Z_\beta$.
Thus $p_{\min}/p_{\max}=e^{-b}$.
Put $A_a(t):=e^{iHt}A_ae^{-iHt}$.
Also set $\mathscr D_t(X):=\sum_a\|[A_a(t),X]\|_{\sigma,\mathrm{GNS}}^2$
and $\Phi_t(Y):=\sum_aA_a(t)^\dagger YA_a(t)$.
The GNS and normalized Hilbert--Schmidt norms satisfy, for every $Y$,
\begin{equation}
 Dp_{\min}\|Y\|_\tau^2\le\|Y\|_{\sigma,\mathrm{GNS}}^2
 \le Dp_{\max}\|Y\|_\tau^2.
 \label{eq:metric}
\end{equation}
Unitary conjugation preserves the normalization $\sum_a A_a^\dagger A_a=I$ and the property of being
closed under adjoints up to phases, and the
family $\{A_a(t)\}$ satisfies the same Poincar\'e inequality because
$\|[A_a(t),O]\|_\tau=\|[A_a,e^{-iHt}Oe^{iHt}]\|_\tau$ and $\Var_\tau$ is
unitarily invariant.  Applying \eqref{eq:generic-L0-Poincare} to
$\{A_a(t)\}$ and \eqref{eq:metric}, and minimizing over scalars (which do
not change $\mathscr D_t$),
\begin{equation}
 \mathscr D_t(X)\ge Dp_{\min}\cdot2\kappa_0\Var_\tau(X)
 \ge2\kappa_0\frac{p_{\min}}{p_{\max}}\Var_{\sigma,\mathrm{GNS}}(X)
 =2e^{-b}\kappa_0\Var_{\sigma,\mathrm{GNS}}(X).                       \label{eq:D-lower}
\end{equation}
For an upper bound, the family $\{A_a(t)\}$ satisfies
$\sum_aA_a(t)^\dagger A_a(t)=I=\sum_aA_a(t)A_a(t)^\dagger$, so
\Cref{lem:commutator-sum} applies to it and gives
\[
 \sum_a\|[A_a(t),Y]\|_\tau^2
 =2\|Y\|_\tau^2-2\Re\langle Y,\Phi_t(Y)\rangle_\tau
 \le4\|Y\|_\tau^2 .
\]
With \eqref{eq:metric} applied to $Y=X-cI$ and minimizing over $c$,
\begin{equation}
 \mathscr D_t(X)\le4e^b\Var_{\sigma,\mathrm{GNS}}(X).                        \label{eq:D-upper}
\end{equation}
Because $[H,\sigma]=0$, the spectral resolution
$H=\sum_E E P_E$ gives
\[
 \|[H,Y]\|_{\sigma,\mathrm{GNS}}^2
 =\sum_{E,E'}\frac{e^{-\beta E}}{Z_\beta}(E'-E)^2
   \|P_{E'}YP_E\|_2^2
 \le W^2\|Y\|_{\sigma,\mathrm{GNS}}^2.
\]
Since $\Tr(\sigma[H,X])=0$ and $[H,X]=[H,X-cI]$,
\[
 \Var_{\sigma,\mathrm{GNS}}([H,X])=\|[H,X]\|_{\sigma,\mathrm{GNS}}^2
 \le W^2\Var_{\sigma,\mathrm{GNS}}(X).
\]
The Jacobi identity
$[[H,A],X]=[H,[A,X]]-[A,[H,X]]$, the triangle inequality in
$\ell^2(\mathcal A)\otimes(\cB(\cH),\|\cdot\|_{\sigma,\mathrm{GNS}})$, and \eqref{eq:D-upper}
applied to $X$ and to $[H,X]$ give
\begin{equation}
 \Bigl(\sum_a\|[[H,A_a(t)],X]\|_{\sigma,\mathrm{GNS}}^2\Bigr)^{1/2}
 \le W \mathscr D_t(X)^{1/2}+\mathscr D_t([H,X])^{1/2}
 \le4e^{b/2}W\Var_{\sigma,\mathrm{GNS}}(X)^{1/2}.                   \label{eq:derivative}
\end{equation}
Let $\mathfrak K_{\gamma,\sigma}$ be the space of equivalence classes,
modulo equality almost everywhere, of measurable families
$Y=(Y_a)_{a\in\mathcal A}$ of functions $Y_a:\mathbb R\to\cB(\cH)$ with
$\sum_a\int_{\mathbb R}\gamma(\omega)\|Y_a(\omega)\|_{\sigma,\mathrm{GNS}}^2d\omega<\infty$,
equipped with the inner product and norm
\[
 \langle Y,Z\rangle_{\mathfrak K_{\gamma,\sigma}}
 :=\sum_{a\in\mathcal A}\int_{\mathbb R}\gamma(\omega)
 \langle Y_a(\omega),Z_a(\omega)\rangle_{\sigma,\mathrm{GNS}}\,d\omega,
 \qquad
 \|Y\|_{\mathfrak K_{\gamma,\sigma}}^2
 :=\langle Y,Y\rangle_{\mathfrak K_{\gamma,\sigma}},
\]
where $\langle X,Y\rangle_{\sigma,\mathrm{GNS}}:=\Tr(\sigma X^\dagger Y)$ is the GNS
inner product.  This is a finite direct sum of weighted $L^2$ spaces of
$\cB(\cH)$-valued functions and hence a Hilbert space.  In it, write
\[
 \begin{gathered}
 [\widehat A_a(\omega),X]
 =\mathscr F_{a,X}(\omega)+\Xi_{a,X}(\omega),
 \qquad
 \mathscr F_{a,X}(\omega):=\widehat f_\beta(\omega)[A_a,X],\\
 \Xi_{a,X}(\omega):=\frac1{\sqrt{2\pi}}\int_{\mathbb R}f_\beta(t)e^{-i\omega t}
 [A_a(t)-A_a,X]\,dt,
 \end{gathered}
\]
and let $\mathscr F_X:=(\mathscr F_{a,X})_{a}$ and
$\Xi_X:=(\Xi_{a,X})_{a}$; both lie in $\mathfrak K_{\gamma,\sigma}$ because
$\widehat f_\beta$ and, by \eqref{eq:bohr}, $\omega\mapsto[\widehat A_a(\omega),X]$
are bounded in norm by $C e^{-c\omega^2}$ for suitable $C,c>0$.  By the definition of $q_\gamma$,
\begin{equation}
 \|\mathscr F_X\|_{\mathfrak K_{\gamma,\sigma}}^2
 =\sum_a\int_{\mathbb R}\gamma(\omega)\widehat f_\beta(\omega)^2d\omega\,
 \|[A_a,X]\|_{\sigma,\mathrm{GNS}}^2
 =q_\gamma\mathscr D_0(X).                                       \label{eq:F-norm}
\end{equation}
For the remainder, $0<\gamma\le1$ and the Plancherel theorem for
Hilbert-space-valued functions, applied to
$t\mapsto f_\beta(t)[A_a(t)-A_a,X]$, whose Fourier transform is
$\Xi_{a,X}$, give
\[
 \|\Xi_X\|_{\mathfrak K_{\gamma,\sigma}}^2
 \le\sum_a\int_{\mathbb R}\|\Xi_{a,X}(\omega)\|_{\sigma,\mathrm{GNS}}^2d\omega
 =\int_{\mathbb R}f_\beta(t)^2\sum_a\|[A_a(t)-A_a,X]\|_{\sigma,\mathrm{GNS}}^2\,dt .
\]
Since $\frac{d}{ds}A_a(s)=i[H,A_a(s)]$, the fundamental theorem of calculus
gives
\[
 [A_a(t)-A_a,X]=i\int_0^t[[H,A_a(s)],X]\,ds,
 \qquad\text{with }\int_0^t:=-\int_t^0\ \text{for }t<0 .
\]
The triangle inequality for the Bochner integral in
$\ell^2(\mathcal A)\otimes(\cB(\cH),\|\cdot\|_{\sigma,\mathrm{GNS}})$ and
\eqref{eq:derivative} give, for every $t$,
\[
\begin{aligned}
 \Bigl(\sum_a\|[A_a(t)-A_a,X]\|_{\sigma,\mathrm{GNS}}^2\Bigr)^{1/2}
 &\le\int_0^{|t|}
 \Bigl(\sum_a\|[[H,A_a(\operatorname{sgn}(t)s)],X]\|_{\sigma,\mathrm{GNS}}^2\Bigr)^{1/2}ds\\
 &\le4e^{b/2}W|t|\,\Var_{\sigma,\mathrm{GNS}}(X)^{1/2}.
\end{aligned}
\]
Inserting this bound and $\int t^2f_\beta(t)^2dt=\beta^2/4$,
\begin{align}
 \|\Xi_X\|_{\mathfrak K_{\gamma,\sigma}}^2
 &\le16e^bW^2\Var_{\sigma,\mathrm{GNS}}(X)
 \int_{\mathbb R}t^2f_\beta(t)^2dt
 =4e^bb^2\Var_{\sigma,\mathrm{GNS}}(X).                             \label{eq:R-bound}
\end{align}
We prove the Poincar\'e inequality \eqref{eq:GNS-poincare} stated below.
If $\Var_{\sigma,\mathrm{GNS}}(X)=0$, then $X$ is scalar and both of its sides vanish,
so assume $\Var_{\sigma,\mathrm{GNS}}(X)>0$.  Then \eqref{eq:D-lower},
\eqref{eq:F-norm}, and \eqref{eq:R-bound} give
\[
 \frac{\|\Xi_X\|_{\mathfrak K_{\gamma,\sigma}}}
      {\|\mathscr F_X\|_{\mathfrak K_{\gamma,\sigma}}}
 \le\frac{2e^{b/2}b\,\Var_{\sigma,\mathrm{GNS}}(X)^{1/2}}
         {\sqrt{2q_\gamma e^{-b}\kappa_0\Var_{\sigma,\mathrm{GNS}}(X)}}
 =\frac{\sqrt2e^bb}{\sqrt{q_\gamma\kappa_0}}
 =\delta_\gamma .
\]
The reverse triangle inequality in $\mathfrak K_{\gamma,\sigma}$ and
\Cref{lem:carre}, whose right side is
$\frac12\|\mathscr F_X+\Xi_X\|_{\mathfrak K_{\gamma,\sigma}}^2$, therefore give
\begin{align}
 -\Re\Tr(\sigma X^\dagger\cL_\gamma(X))
 &\ge\frac12q_\gamma(1-\delta_\gamma)^2\mathscr D_0(X)
 \ge q_\gamma(1-\delta_\gamma)^2e^{-b}\kappa_0
 \Var_{\sigma,\mathrm{GNS}}(X).                                    \label{eq:GNS-poincare}
\end{align}
This inequality first proves primitivity: a fixed operator has zero
GNS variance and is therefore scalar, since $\sigma$ is full rank.  By
\eqref{eq:KMS-DB} and \Cref{lem:gap-primitive}, $\cL_\gamma$ has a complete
KMS-orthonormal eigenbasis with real eigenvalues.  If $\cL_\gamma(X)=-\lambda X$ with $\lambda\ne0$,
then KMS-orthogonality of $X$ to $I$ gives $\Tr(\sigma X)=0$, so
$\Var_{\sigma,\mathrm{GNS}}(X)=\|X\|_{\sigma,\mathrm{GNS}}^2$, and \eqref{eq:GNS-poincare} gives
$\lambda\ge q_\gamma(1-\delta_\gamma)^2e^{-b}\kappa_0$.  Since
$\gap(-\cL_\gamma)$ is the smallest nonzero eigenvalue of $-\cL_\gamma$
(\Cref{def:gap}), this proves \eqref{eq:comparison-gap}.  The monotonicity
and substitution statements follow from the identity
\[
 q_\gamma e^{-b}\kappa_0(1-\delta_\gamma)^2
 =q_\gamma e^{-b}
 \left(\sqrt{\kappa_0}-\sqrt{\frac2{q_\gamma}}\,e^bb\right)^2.
\]
On the region $\delta_\gamma<1$, the expression in parentheses is positive;
it increases with $\kappa_0$ and decreases with $b$, while $e^{-b}$ also
decreases with $b$.  Moreover, the defining hypotheses for
$\bar\delta_\gamma$ imply
$\delta_\gamma\le\bar\delta_\gamma<1$.  Replacing $\kappa_0$ by
$\underline\kappa$ and $b$ by $\bar b$ therefore gives
\eqref{eq:comparison-substituted-gap}.
\end{proof}

\begin{addedremark}[Spectral gap stability]
For fixed couplings, uniform bounds $\kappa_0\ge\kappa_*>0$ and
$\beta W\le B$ give a uniform gap
$\gap(-\cL_{\beta,\mathrm M})\ge\tfrac12e^{-B}\kappa_*>0$.
Near infinite temperature, set
$g_{\gamma,0}:=q_\gamma\gap(-\cL_{\mathcal A,0})>0$,
the spectral gap of the CKG generator at infinite temperature in \eqref{eq:infinite temperature-limit}.
Taking $\kappa_0=\gap(-\cL_{\mathcal A,0})$ in
\Cref{thm:comparison}, with $b=\beta W$, gives
\[
 \frac{\gap(-\cL_\gamma)}{g_{\gamma,0}}
 \ge e^{-b}(1-\delta_\gamma)^2,
 \qquad
 \delta_\gamma=\frac{\sqrt2e^bb}{\sqrt{g_{\gamma,0}}}<1.
\]
The right-hand side tends to one as $b\to0$ for fixed $g_{\gamma,0}$.
This is a one-sided stability bound whose range of admissible inverse
temperatures can shrink with system size. Perturbative stability also underlies
high-temperature CKG Lindbladian spectral gap analysis \cite{RouzeFrancaAlhambra2024}.
\end{addedremark}

\subsection{Combined bounds and transfer from a Markov chain}

\begin{remark}[Combining the two spectral gap estimates]
If, in addition,
$\delta_{\gamma_{\rm M}}<1$, then \Cref{thm:comparison} and
\eqref{eq:metropolis-all-temperature-gap} may be combined as
\begin{equation}
 \gap(-\cL_{\gamma_{\rm M}})
 \ge e^{-b}\kappa_0
 \max\left\{\frac12,
 q_{\gamma_{\rm M}}(1-\delta_{\gamma_{\rm M}})^2\right\}.
 \label{eq:combined-metropolis-gap}
\end{equation}
\end{remark}

\begin{remark}[Scalar Hamiltonians]
\label{rem:scalar-H}
If $H$ is scalar, then $b=0$, $\widehat A_a(\omega)=\widehat f_\beta(\omega)A_a$,
$B_\gamma=0$, and $\cL_\gamma=q_\gamma\cL_{\mathcal A,0}$.  Hence
$\gap(-\cL_\gamma)=q_\gamma\gap(-\cL_{\mathcal A,0})$.
The comparison bound \eqref{eq:comparison-gap} is therefore an equality at
$b=0$ precisely when $\kappa_0$ is the optimal Poincar\'e constant of the
normalized infinite-temperature generator.  The factors $e^{\pm b}$ in
\eqref{eq:D-lower}--\eqref{eq:D-upper} arise from the comparison of the
GNS and normalized Hilbert--Schmidt norms in \eqref{eq:metric}; these
factors are not claimed to be optimal at finite temperature.
\end{remark}

\begin{corollary}[Transfer from a symmetric Markov chain]
\label{cor:generic}
Fix $\beta>0$ and $\gamma\in\{\gamma_{\rm G},\gamma_{\rm M}\}$.  Let $F$ be
a finite set with $|F|\ge2$, let
$Q=Q^{\mathsf T}$ be the transition matrix of an irreducible Markov chain
on $F$, let $\{A_a\}$ be any Kraus family for $\Phi_Q$ that is closed under
adjoints up to phases
(\Cref{def:Markov-kernel-embedding} and \eqref{eq:Hermitian-jumps}), and let
$H$ on $\ell^2(F)$ have spectral width $W$. For $\gamma=\gamma_{\rm M}$,
the generator is primitive at every $\beta>0$, has the unique Gibbs
stationary state, and satisfies
$\gap(-\cL_{\gamma_{\rm M}})\ge\tfrac12e^{-\beta W}\lambda_Q$.
For either weight, if
\begin{equation}
 \sqrt2e^{\beta W}\beta W<\sqrt{q_\gamma\lambda_Q},       \label{eq:generic-window}
\end{equation}
then the exact CKG generator is primitive and
\[
 \gap(-\cL_\gamma)\ge
 q_\gamma e^{-\beta W}\lambda_Q
 \left(1-
 \frac{\sqrt2e^{\beta W}\beta W}
 {\sqrt{q_\gamma\lambda_Q}}\right)^2.
\]
\eqref{eq:generic-window} is equivalent to
$\beta W<\mathrm W_0(\sqrt{q_\gamma\lambda_Q/2})$,
where $\mathrm W_0(x)$ is the unique nonnegative number $w$ satisfying
$we^w=x$ for $x\ge0$; this is the principal branch of the Lambert $W$
function \cite{CorlessEtAl1996}.
\end{corollary}

\begin{proof}
\Cref{lem:carre-L0} gives \eqref{eq:generic-L0-Poincare} with
$\kappa_0=\lambda_Q$; apply \Cref{thm:metropolis-all-temperature,thm:comparison}. The Lambert form
follows because $b\mapsto be^b$ is increasing on $[0,\infty)$.
\end{proof}

\section{Trace-distance mixing}
\label{sec:mixing}

The gap estimates in \Cref{sec:comparison-main} give trace-distance
mixing bounds through the quantum $\chi^2$ estimate. For a primitive KMS-symmetric generator $\cL$ with spectral gap
$\lambda:=\gap(-\cL)>0$, its \emph{relaxation time} is
$t_{\rm rel}:=\lambda^{-1}$. The following is the continuous-time, $\alpha=1/2$ specialization of the
quantum $\chi^2$ mixing bound in
\cite[Lemmas~5 and~12]{TemmeEtAl2010}; see also
\cite[Eq.~(2)]{KastoryanoTemme2013}.

\begin{lemma}[Mixing from the KMS spectral gap]
\label{lem:mixing}
Let $T_t=e^{t\cL}$ be primitive and KMS symmetric with full-rank
stationary state $\sigma$ on a $D$-dimensional space, and let
$0<g\le\gap(-\cL)$.  Then every state $\rho_0$ satisfies
\begin{equation}
 \frac12\|T_{t,*}(\rho_0)-\sigma\|_1
 \le\frac12e^{-gt}
 \sqrt{\chi^2_{1/2}(\rho_0\|\sigma)},                \label{eq:mixing}
\end{equation}
where
$\chi^2_{1/2}(\rho\|\sigma):=
 \Tr[(\sigma^{-1/4}\rho\sigma^{-1/4})^2]-1$ is the symmetric quantum
$\chi^2$-divergence.  Moreover,
$\sup_{\rho_0}\chi^2_{1/2}(\rho_0\|\sigma)=\sigma_{\min}^{-1}-1$, where
$\sigma_{\min}:=\lambda_{\min}(\sigma)$.
For $0<\epsilon<1$, define
$t_{\rm mix}(\epsilon):=
 \inf\{t\ge0:\sup_{\rho_0}\frac12\|T_{t,*}(\rho_0)-\sigma\|_1\le\epsilon\}$.
For the Gibbs state $\sigma=\sigma_\beta$ of a Hamiltonian of spectral width $W$ on a
$D$-dimensional space, put $[x]_+:=\max\{x,0\}$.  Then
\begin{equation}
 t_{\rm mix}(\epsilon)\le
 \frac1g\left[
 \log\frac1{2\epsilon}
 +\frac12\log\max\{1,De^{\beta W}-1\}
 \right]_+.                                               \label{eq:mixing-time}
\end{equation}
\end{lemma}

\begin{proof}
Equation~\eqref{eq:mixing} follows from
\cite[Lemmas~5 and~12]{TemmeEtAl2010} applied to $T_t$ with
$\alpha=1/2$. For the exact worst-case prefactor, convexity reduces the
supremum to pure states, where
$\chi^2_{1/2}(|\psi\rangle\langle\psi|\|\sigma)+1
=\langle\psi|\sigma^{-1/2}|\psi\rangle^2\le\sigma_{\min}^{-1}$,
with equality at a minimum-eigenvalue eigenvector.
Finally $\sigma_{\min}\ge e^{-\beta W}/D$ gives \eqref{eq:mixing-time}.
\end{proof}

For $D\ge2$ and $0<\epsilon<1/2$, the actual gap gives the explicit relation
\begin{equation}
 t_{\rm rel}\log\frac1{2\epsilon}
 \le t_{\rm mix}(\epsilon)
 \le t_{\rm rel}\left(
 \log\frac1{2\epsilon}+\frac12\log(\sigma_{\min}^{-1}-1)\right).
 \label{eq:relaxation-mixing-relation}
\end{equation}
The upper bound follows from \Cref{lem:mixing} with $g=\lambda$.
For the lower bound, choose a Hermitian eigenobservable $A$ with
$\cL(A)=-\lambda A$ and $\|A\|=1$. Stationarity implies
$\Tr(\sigma A)=0$. An eigenstate $\rho$ with
$|\Tr(\rho A)|=1$ then satisfies
$\tfrac12\|T_{t,*}(\rho)-\sigma\|_1\ge\tfrac12e^{-\lambda t}$
by trace-norm duality. This proves the lower bound.

\subsection{Polynomial mixing for logarithmic \texorpdfstring{$\beta W$}{beta W}}

For this consequence, $\cL_\beta$ denotes the CKG
generator with transition weight $\gamma_{\rm M}$ defined in \Cref{sec:ckg}.

\begin{corollary}[Mixing when $\beta W$ grows logarithmically]
\label{appth:cor:log-temperature}
Consider a family satisfying \Cref{ass:poincare}, indexed by integers
$n\ge2$. Suppose $\kappa_0\ge n^{-a}$ and $\beta W\le c\log n$ for
constants $a,c\ge0$ independent of $n$. Then
\begin{equation}
 \gap(-\cL_\beta)\ge\tfrac12 n^{-(a+c)},\qquad
 t_{\rm mix}(\epsilon)\le
 2n^{a+c}\left[
 \log\frac1\epsilon+\frac12\log D+\frac c2\log n
 \right]\quad(0<\epsilon<1).
 \label{appth:eq:log-temperature}
\end{equation}
In particular, if $\log D$ is polynomially bounded in $n$, this gives
polynomial semigroup mixing time in $n$ and $\log(1/\epsilon)$.
\end{corollary}

\begin{proof}
Combine \Cref{thm:metropolis-all-temperature,lem:mixing} with
$e^{-\beta W}\kappa_0\ge n^{-(a+c)}$ and
$\lambda_{\min}(\sigma_\beta)\ge e^{-\beta W}/D$, using the mixing-time upper bound $g^{-1}[\log(1/\epsilon)+(\log D+\beta W)/2]$.
\end{proof}

For example, if $W=1$ and $c>0$, choosing $\beta=c\log n$ corresponds
to temperature $1/(c\log n)$ in units with Boltzmann constant one.

\begin{addedremark}[Limitations of ground-space preparation]
Let $P_0$ be the ground-space projection and set
$E_0=\lambda_{\min}(H)$.
In general,
$\Tr(P_0\sigma_\beta)\le e^{\beta W}\operatorname{rank}(P_0)/D$,
since $\Tr(e^{-\beta(H-\lambda_{\min}(H)I)})\ge De^{-\beta W}$.
Thus, if $\operatorname{rank}(P_0)/D=e^{-\Omega(n)}$ and
$\beta W=O(\log n)$, the Gibbs weight in the ground state subspace tends to zero.
This also applies to the harmonic subspace of a Hodge Laplacian.

If $E_0$ is known,
the ideal energy filter $F=e^{-\beta(H-E_0I)/2}$ applied to $I/D$
produces the Gibbs state upon success, with probability
$p_{\rm succ}=D^{-1}\Tr(F^\dagger F)\ge e^{-\beta W}$.
Logarithmic $\beta W$ therefore also permits polynomial repetition
overhead for established energy-filtering approaches
\cite[Sec.~III]{ChowdhurySomma2017}.
Efficient filter implementation remains a separate requirement.
Our comparison establishes convergence of the CKG dynamics with transition weight $\gamma_{\rm M}$,
without asserting an advantage over those algorithms.
\end{addedremark}

\section{Temperature dependence and upper bounds on the gap}
\label{app:sharpness}

The rate calculation below is used in the example of
\Cref{sec:comparison-main}. It also gives the
spectral gap upper bounds that establish the limits of the general comparison. 
Unless stated otherwise, $\beta>0$, the transition weight is
$\gamma_{\rm M}$ in \eqref{eq:weights}, and
$\cL_\beta:=\cL_{\gamma_{\rm M}}$ denotes the exact
Heisenberg-picture generator for the Hamiltonian and couplings
under consideration. Its trace
dual is $\cL_{\beta,*}$. We use
$D_{\rm tr}(\rho,\sigma):=\frac12\|\rho-\sigma\|_1$ for trace distance.
The results concern these chosen dynamics. An exponentially small gap
does not imply that their Gibbs state is difficult to prepare by another
method.

\subsection{Scalar transition rate after Gaussian filtering}

Write $\Phi_{\mathrm N}(x):=(2\pi)^{-1/2}\int_{-\infty}^x e^{-z^2/2}\,dz$
for the standard normal distribution function, and set
$q:=q_{\gamma_{\rm M}}=2\Phi_{\mathrm N}(-1/2)$.

The scalar transition rate below is obtained by Gaussian smoothing of
the transition weight $\gamma_{\rm M}$; see
\cite[Sec.~4.3.2, Eq.~(44)]{GCDK2024}, expressed in the
normalization of \Cref{sec:ckg}. We recall the formula and verify
the elementary bounds needed for the comparisons.

\begin{lemma}[Scalar transition rate after Gaussian filtering]
\label{appth:lem:rate}
If $[H,A]=\nu A$, then
$\int_{\mathbb R}\gamma_{\rm M}(\omega)
\widehat f_\beta(\omega-\nu)^2\,d\omega=r(\beta\nu)$, where
\begin{equation}
 r(x)=\Phi_{\mathrm N}(-x-\tfrac12)
       +e^{-x}\Phi_{\mathrm N}(x-\tfrac12).
 \label{appth:eq:rate}
\end{equation}
The function is decreasing on $\mathbb R$, $r(0)=q$, and for $x\ge0$,
\begin{equation}
 r(x)=e^{-x}r(-x),\qquad
 qe^{-x}\le r(x)\le e^{-x},\qquad
 r(x)+r(-x)\ge1,\qquad
 \lim_{x\to\infty}r(-x)=1.
 \label{appth:eq:ratebounds}
\end{equation}
\end{lemma}

\begin{proof}
For $[H,A]=\nu A$, the filtered operator is
$\widehat A(\omega)=\widehat f_\beta(\omega-\nu)A$.
With $x=\beta\nu$ and $z=\beta\omega$, its scalar rate is
\[
 r(x)=\int_{\mathbb R}e^{-\max\{z+1/2,0\}}
             \frac{e^{-(z-x)^2/2}}{\sqrt{2\pi}}\,dz.
\]
Equation~\eqref{appth:eq:rate} is \cite[Eq.~(44)]{GCDK2024} with Gaussian
width one. Differentiation gives
$r'(x)=-e^{-x}\Phi_{\mathrm N}(x-1/2)<0$, while substitution gives
$r(x)=e^{-x}r(-x)$ and $r(0)=q$.
The integral is at most one, so $q\le r(-x)\le1$ for $x\ge0$,
which proves the two-sided bound. For the sum bound, change $z$ to $-z$
in $r(-x)$ and use
$e^{-\max\{z+1/2,0\}}+e^{-\max\{-z+1/2,0\}}\ge1$.
At $x=0$, this inequality is strict on an interval of positive Gaussian
measure, so $q>1/2$. Finally,
$r(-x)=\mathbb E[e^{-\max\{Z-x+1/2,0\}}]\to1$ by dominated
convergence for a standard normal random variable $Z$.
\end{proof}

\begin{proposition}[Population and coherence blocks]
\label{appth:prop:energybasis}
Let $H=\sum_{i=1}^D E_i\ket i\bra i$ in an orthonormal basis, let $Q$ be a
symmetric irreducible stochastic matrix on $\{1,\ldots,D\}$, and use
$K_{ij}=\sqrt{Q(i,j)}\ket j\bra i$. The diagonal restriction of
$\cL_\beta$ is the reversible classical generator $L_{\rm pop}$ with
transition rates $Q(i,j)r(\beta(E_j-E_i))$ for $j\ne i$.
If $s_i:=\sum_jQ(i,j)r(\beta(E_j-E_i))$, each off-diagonal matrix unit
$\ket i\bra j$ has decay rate $(s_i+s_j)/2$. In particular,
\begin{equation}
 \gap(-\cL_\beta)
 =\min\left\{\gap(-L_{\rm pop}),\
                 \min_{i\ne j}\frac{s_i+s_j}{2}\right\}.
 \label{appth:eq:blocks}
\end{equation}
\end{proposition}

\begin{proof}
Each coupling has a single Bohr frequency, so \Cref{appth:lem:rate}
applies separately. The loss operator $R_{\gamma_{\rm M}}$ is diagonal
with entries $s_i$. It commutes with $H$, and the coherent correction in
\eqref{eq:CKG-data} therefore vanishes. The gain map annihilates every
off-diagonal matrix unit. On a diagonal observable with entries $x_i$,
the generator has entries
$(L_{\rm pop}x)_i=\sum_{j\ne i}Q(i,j)r(\beta(E_j-E_i))(x_j-x_i)$.
Self-loops cancel in this expression but remain in $s_i$.
Writing $\pi_i=e^{-\beta E_i}/\sum_j e^{-\beta E_j}$,
\eqref{appth:eq:ratebounds} gives
$\pi_iQ(i,j)r(\beta(E_j-E_i))
=\pi_jQ(j,i)r(\beta(E_i-E_j))$.
Thus $L_{\rm pop}$ is reversible and irreducible, and its nonzero decay
rates are positive. Each $s_i$ is positive as well. The invariant-block
decomposition proves \eqref{appth:eq:blocks}. The energy eigenvalues need not be distinct. Decompositions into population and coherence subspaces are also used
in spectral gap analyses of quantum thermalization; see \cite{Temme2013}.
\end{proof}

\subsection{Optimality of the exponential dependence on \texorpdfstring{$\beta W$}{beta W}}\label{app:optimal}

The example below has gap asymptotic to $\kappa_0e^{-\beta W}$ as
$\beta W\to\infty$. It rules out replacing $e^{-\beta W}$ in the
general lower bound by $e^{-\eta\beta W}$ with any $\eta<1$ and a
positive universal prefactor. Thus the coefficient of $\beta W$ in
the exponent is optimal under the stated assumptions. The optimal
numerical prefactor remains undetermined.

\begin{theorem}[Exact spectral gap of the CKG generator for a three-dimensional system]
\label{appth:thm:sharpness}
Fix $a=1/4$ and $W>0$. In the orthonormal basis
$\ket0,\ket1,\ket2$, set
\begin{equation}
 Q=\begin{pmatrix}1-a&a&0\\a&1-2a&a\\0&a&1-a\end{pmatrix},
 \qquad H=\operatorname{diag}(0,W,0),
 \qquad K_{ij}=\sqrt{Q(i,j)}\ket j\bra i.
 \label{appth:eq:barrier}
\end{equation}
The normalized infinite-temperature generator $\Phi_Q-\id$ has gap
$\kappa_0=a$, and the spectral gap of the exact CKG generator is
\begin{equation}
 \gap(-\cL_\beta)=\kappa_0r(\beta W),\qquad
 \lim_{\beta W\to\infty}
 \frac{\gap(-\cL_\beta)}{\kappa_0e^{-\beta W}}=1.
 \label{appth:eq:sharpness}
\end{equation}
\end{theorem}

\begin{proof}
The eigenvalues of $Q$ are $1,1-a,1-3a$, so \Cref{thm:associated-gap}
gives $\kappa_0=a$. Put $u=r(\beta W)$ and $v=r(-\beta W)$.
By \Cref{appth:prop:energybasis}, the population generator on observables is
\[
 L_{\rm pop}=
 \begin{pmatrix}-au&au&0\\av&-2av&av\\0&au&-au\end{pmatrix},
\]
with eigenvalues $0,-au,-a(u+2v)$. The diagonal entries of the loss operator are
$s_0=s_2=(1-a)q+au$ and $s_1=(1-2a)q+2av$.
Since $v\ge q$ and $u\le q$, one has $s_1\ge q$ and
$s_0=s_2\ge3q/4$. All off-diagonal decay rates therefore exceed
$au\le q/4$. The spectral gap on the full operator space is $au$, and \Cref{appth:lem:rate} gives the
limit in \eqref{appth:eq:sharpness}.
\end{proof}

\begin{corollary}[Optimal exponent and bounds on the prefactor]
\label{appth:cor:necessity}
No bound $\gap(-\cL_\beta)\ge c\kappa_0e^{-\eta\beta W}$ with universal
constants $c>0$ and $\eta<1$ can hold for all inputs of
\Cref{thm:metropolis-all-temperature}. In particular, a universal
inverse-polynomial replacement for the exponential dependence on $\beta W$
is impossible. If $C_*$ is the largest universal constant for which
$\gap(-\cL_\beta)\ge C_*e^{-\beta W}\kappa_0$ holds when $\kappa_0$
is the optimal infinite-temperature Poincar\'e constant, then $1/2\le C_*\le q$.
\end{corollary}

\begin{proof}
The first assertions follow by letting $\beta W\to\infty$ in
\eqref{appth:eq:sharpness}. The lower bound on $C_*$ follows from
\Cref{thm:metropolis-all-temperature}. If $H$ is scalar, then $W=0$
and the CKG generator is $q\cL_{\mathcal A,0}$, as in \Cref{rem:scalar-H}.
Its gap equals $q\kappa_0$, proving the upper bound. The optimal numerical
prefactor is not determined here.
\end{proof}

The same example gives a mixing-time lower bound for a specified initial state.
In the three-dimensional system example, label the basis by $0,1,2$ and take
$X=\operatorname{diag}(1,0,-1)$. This norm-one eigenobservable has zero
Gibbs mean, and $\cL_\beta(X)=-a r(\beta W)X$. Consequently,
\begin{equation}
 D_{\rm tr}\!\left(
 e^{t\cL_{\beta,*}}(\ket0\bra0),\sigma_\beta\right)
 \ge\frac12e^{-a r(\beta W)t}.
 \label{appth:eq:seed-lower-bound}
\end{equation}
Indeed, the difference of the two expectations of $X$ is
$e^{-a r(\beta W)t}$, and trace-norm duality bounds this difference by
twice the trace distance. Since $r(\beta W)\le e^{-\beta W}$, trace-distance error at most $\epsilon<1/2$ requires at least
$a^{-1}e^{\beta W}\log(1/(2\epsilon))$ units of evolution time.
This is a lower bound for the chosen dynamics, not for other
preparation algorithms.

\subsection{Spectral gap upper bound at low temperature}

Let $P_0$ be the orthogonal projection onto the ground eigenspace of $H$.
The \emph{projected couplings} are the operators $P_0A_aP_0$ acting on
$\operatorname{ran}P_0$. We call this family \emph{reducible} if there
is an orthogonal projection $P\le P_0$, distinct from $0$ and $P_0$,
that commutes with every projected coupling.
Equivalently, all projected couplings preserve
both $\operatorname{ran}P$ and its orthogonal complement within the
ground state subspace. Under this condition, the next theorem gives an upper
bound on the spectral gap of the CKG generator that tends to zero exponentially as
$\beta\to\infty$, with the Hamiltonian and couplings fixed.
It does not identify the exact asymptotic decay rate.
For a Hodge Laplacian with nonzero kernel, the ground state subspace is the
harmonic subspace.

\begin{theorem}[Spectral gap upper bound for reducible projected couplings]
\label{appth:thm:groundobstruction}
Adopt \Cref{ass:poincare}, shift the ground state energy of $H$ to zero, and
suppose its ground-space projection $P_0$ has rank $h\ge2$ and its positive
spectrum is nonempty, with minimum $\delta_H>0$. If an orthogonal
projection $P\le P_0$ of rank $r$, where $0<r<h$, commutes with every
$P_0A_aP_0$, then
\begin{equation}
 \gap(-\cL_\beta)\le\frac{h}{h-r}e^{-\beta\delta_H}.
 \label{appth:eq:groundobstruction}
\end{equation}
\end{theorem}

\begin{proof}
Write $Z=\Tr(e^{-\beta H})$ and $p_0=1/Z$. Since $P$ commutes with
$\sigma_\beta$, its KMS Dirichlet form
$\cE_\beta(P):=-\langle P,\cL_\beta(P)\rangle_{\sigma_\beta,1/2}$
equals $-\Re\Tr(\sigma_\beta P\cL_\beta(P))$.
Its KMS variance is $rp_0(1-rp_0)$. Choose an orthonormal energy basis
that diagonalizes $P$ within the ground state subspace, and let
$u_1,\ldots,u_r$ be its basis vectors spanning $\operatorname{ran}P$.
The commutator with $P$ has entries only between its range and its
complement. Entries within the ground state subspace vanish by hypothesis.

Apply the stationary commutator identity of \Cref{lem:carre}.
In the energy basis, taking the weighted Hilbert--Schmidt norm removes
cross terms between different matrix entries; each matrix entry has one
Bohr frequency. The terms directed from $u_\ell$ to a positive-energy
vector $e_j$ have weight $p_0r(\beta E_j)$, and the reverse terms have
weight $p_0e^{-\beta E_j}r(-\beta E_j)$. Relabeling the adjoint couplings
identifies their squared matrix elements, and detailed balance in
\eqref{appth:eq:ratebounds} identifies their weights. The factor $1/2$
in \Cref{lem:carre} therefore gives
\[
 \cE_\beta(P)
 =p_0\sum_{\ell=1}^r\sum_{j:E_j>0}
       r(\beta E_j)\sum_a|\bra{e_j}A_a\ket{u_\ell}|^2
 \le rp_0e^{-\beta\delta_H}.
\]
Here the last inequality uses $r(x)\le e^{-x}$ and
$\sum_aA_a^\dagger A_a=I$. The generator is primitive by
\Cref{thm:metropolis-all-temperature}, so the KMS variational principle
bounds its gap by $e^{-\beta\delta_H}/(1-rp_0)$.
Since $Z\ge h$, this is at most \eqref{appth:eq:groundobstruction}.
The proof does not assume that the coherent correction vanishes.
\end{proof}

For a fixed Hamiltonian with a degenerate ground state subspace and fixed couplings,
irreducibility of the projected couplings is therefore necessary for a
gap bounded uniformly away from zero as $\beta\to\infty$. Here
irreducibility means that the only orthogonal projections on the ground
space commuting with every $P_0A_aP_0$ are $0$ and $P_0$.
An unweighted Hodge Laplacian with reducible projected couplings is given in \Cref{app:hodge-barbell-example}.

\section{Hodge Laplacians and the uniform down--up walk}
\label{app:hodge-details}

This appendix gives the definitions and proof used in
\Cref{sec:hodge-application}.
Bounds based on local spectral expansion and explicit examples follow
in \Cref{app:hodge-support}.

\subsection{The Hodge Laplacian as a Hamiltonian}
\label{sec:model}

We use the simplicial-chain and unweighted Hodge Laplacian conventions of
\cite[Sec.~2, Definitions~2.1--2.3]{HorakJost2013}.

\begin{assumption}[Simplicial complex setup]
\label{set:complex}
Use the finite unweighted simplicial complex $X$, face dimension $k$, face set
$F=X(k)$, cardinality $D_k\ge2$, and increasing-vertex orientation of
\Cref{sec:hodge-application}. The complex need not be pure.
\end{assumption}

For $j\ge0$, let $\cH_j:=\ell^2(X(j))$, where $\ket f$ is the indicator
function of $f\in X(j)$; these vectors form the orthonormal face basis.
For $j\ge1$, the boundary map $\partial_j:\cH_j\to\cH_{j-1}$ acts as
$\partial_j\ket f=\sum_{i=0}^j(-1)^i\ket{f\setminus\{x_i\}}$
when $f=\{x_0<\cdots<x_j\}$. Choose nonnegative bounds $\alpha_j\ge\|\partial_j\|$ for $j=k,k+1$,
and define
\begin{equation}
 \Lambda_k:=\alpha_k^2+\alpha_{k+1}^2,\qquad\Delta_k=\partial_k^\dagger\partial_k+\partial_{k+1}\partial_{k+1}^\dagger,\qquad H_{X,k}:=\frac{\Delta_k}{\Lambda_k},
 \label{eq:normalized-H}
\end{equation}

Write $\Delta_k^\downarrow:=\partial_k^\dagger\partial_k$ and
$\Delta_k^\uparrow:=\partial_{k+1}\partial_{k+1}^\dagger$ for the lower
and upper Hodge Laplacians, respectively.
The identity
$\partial_k\partial_{k+1}=0$ implies
$\Delta_k^\downarrow\Delta_k^\uparrow=
\Delta_k^\uparrow\Delta_k^\downarrow=0$.  Thus the two positive
semidefinite summands have orthogonal ranges. The nonzero spectrum of
$\Delta_k$ is the union of their nonzero spectra, counting multiplicities; see
\cite[Sec.~2, Eqs.~(2.3)--(2.5)]{HorakJost2013}.  Consequently,
\begin{equation}
 \|\Delta_k\|=\max\{\|\partial_k\|^2,\|\partial_{k+1}\|^2\}.
 \label{eq:hodge-norm}
\end{equation}

\begin{lemma}[Norm bounds for the boundary operators~\cite{HorakJost2013}]
\label{lem:boundary-norms}
Adopt \Cref{set:complex}.  Suppose that $R_-,R_+\ge0$ satisfy
\begin{equation}
 \max_{\xi\in X(k-1)}\bigl|\{f\in X(k):\xi\subset f\}\bigr|\le R_-,
 \quad
 \max_{f\in X(k)}\bigl|\{\zeta\in X(k+1):f\subset\zeta\}\bigr|\le R_+ .
 \label{eq:degree-bounds}
\end{equation}
Then
\begin{equation}
 \|\partial_k\|^2\le(k+1)R_-,
 \quad
 \|\partial_{k+1}\|^2\le(k+2)R_+,
 \quad
 \|\Delta_k\|\le\max\{(k+1)R_-,\ (k+2)R_+\}.
 \label{eq:boundary-norms}
\end{equation}
The choices $R_-=n-k$ and $R_+=n-k-1$ are always admissible, so
$\|\Delta_k\|\le(k+2)(n-k)$ unconditionally.
\end{lemma}

\begin{proof}
The degree bound for the upper Hodge Laplacian in
\cite[Eq.~(3.9), arXiv v3]{HorakJost2013}, with indices $i=k-1$ and
$i=k$, gives the first two bounds, using the equality of the nonzero spectra of the upper Hodge Laplacian
in degree $i$ and the lower Hodge Laplacian in degree $i+1$ in \cite[Eq.~(2.6)]{HorakJost2013}.
If $X(k+1)=\emptyset$, then $\partial_{k+1}=0$.
The full Hodge Laplacian bound follows from \eqref{eq:hodge-norm}.
Each extension adds one of at most $n-k$ or $n-k-1$ remaining vertices,
respectively, and
$\max\{(k+1)(n-k),(k+2)(n-k-1)\}\le(k+2)(n-k)$.
\end{proof}

By \eqref{eq:hodge-norm}, the spectral width $W$ of $H_{X,k}$ also satisfies
$W\le\max\{\alpha_k^2,\alpha_{k+1}^2\}/\Lambda_k\le1$.

\subsection{The uniform down--up walk and its gap bound}
\label{sec:walk}
\label{subsec:walk-definition}

Let $Y:=\langle F\rangle:=\{\alpha:\alpha\subseteq f\text{ for some }f\in F\}$ be the
pure $k$-dimensional complex generated by $F=X(k)$.
A \emph{ridge} is a face $\eta\in Y(k-1)$.
Write $S_\eta:=\{v\in[n]\setminus\eta:\eta\cup\{v\}\in F\}$ and
$d_\eta:=|S_\eta|=|\{f\in F:\eta\subseteq f\}|\ge1$.

\begin{definition}[Uniform down--up walk]
\label{def:Pu}
Under \Cref{set:complex}, the uniform down--up walk deletes a uniformly
chosen vertex of the current face and then selects a uniformly chosen
$k$-face containing the resulting ridge. Its transition matrix is
\begin{equation}
 P_u(f,g):=\frac1{k+1}
 \sum_{\substack{\eta\in Y(k-1)\\ \eta\subseteq f\cap g}}\frac1{d_\eta},
 \qquad f,g\in F.
 \label{eq:Pu}
\end{equation}
\end{definition}
This is the down--up walk induced by the uniform distribution on $Y(k)$
in \cite[Sec.~2.4]{AlevLau2020}.

\subsubsection{Structure and Dirichlet form of the walk}
\label{app:uniform-walk}

The Dirichlet form $\cE_{u,P_u}$ and variance $\Var_u$ use the
conventions of \Cref{subsec:poincare} with $u_f=D_k^{-1}$.

Parts (i) and (ii) of the next lemma specialize the down--up factorization
and local restriction identities to the uniform distribution on $Y(k)$; see \cite[Sec.~2.4 and Lemma~3.7]{AlevLau2020}.
\begin{lemma}[Properties of the uniform down--up walk]
\label{lem:Pu-structure}
Let $u_f:=D_k^{-1}$, let
$B_{\eta f}:=\mathbf 1_{\{\eta\subseteq f\}}$, and set
$\mathsf N:=\operatorname{diag}(d_\eta)_{\eta\in Y(k-1)}$.  Under \Cref{def:Pu}, the following statements hold.
\begin{enumerate}
\item[(i)]
The matrix $\mathsf N$ is invertible, and
$P_u=(k+1)^{-1}B^{\mathsf T}\mathsf N^{-1}B$.
Consequently, $P_u$ is a symmetric, positive-semidefinite Markov
transition matrix that is reversible with respect to $u$.

\item[(ii)]
For a ridge $\eta$ and $a\in\mathbb C^F$, write
$\bar a_\eta:=d_\eta^{-1}\sum_{f\supseteq\eta}a_f$ and
$\Var_\eta(a):=d_\eta^{-1}\sum_{f\supseteq\eta}|a_f-\bar a_\eta|^2$ for
the mean and variance of $a$ under the uniform distribution on the $k$-faces containing $\eta$. The Dirichlet form of the walk satisfies
\begin{equation}
 \cE_{u,P_u}(a)
 =\frac1{D_k(k+1)}
 \sum_{\eta\in Y(k-1)}d_\eta\Var_\eta(a).
 \label{eq:Pu-form}
\end{equation}

\item[(iii)]
For distinct $f,g\in F$, one has $P_u(f,g)>0$ if and only if
$f\cap g\in Y(k-1)$.  Hence the support graph of $P_u$ is the
ridge-adjacency graph on $F$.  Moreover,
$\operatorname{spec}(P_u)\subseteq[0,1]$, and $P_u$ is irreducible if and
only if the ridge-adjacency graph is connected, equivalently, if and only
if $1-\lambda_2(P_u)>0$.
\end{enumerate}
\end{lemma}

\begin{proof}
The factorization in (i) follows entrywise from \eqref{eq:Pu}; it gives
symmetry and positive semidefiniteness. Each face has $k+1$ ridges, so
$\sum_gP_u(f,g)=(k+1)^{-1}\sum_{\eta\subseteq f}d_\eta/d_\eta=1$,
and symmetry gives reversibility with respect to $u$.
For (ii), the induced ridge distribution is
$d_\eta/[D_k(k+1)]$, and conditional on a ridge $\eta$, its containing $k$-faces are uniformly
distributed. Subtracting the local restriction identities in
\cite[Lemma~3.7(1)--(2)]{AlevLau2020} therefore gives
\eqref{eq:Pu-form}. For complex $a$, apply these identities to its real
and imaginary parts. Finally, (iii) follows from \eqref{eq:Pu},
\Cref{eq:gap-irreducible}, and positive semidefiniteness.
\end{proof}

\subsection{Proof of the CKG Lindbladian spectral gap bound for the Hodge Laplacian}

\begin{proof}[Proof of \Cref{thm:main}]
Symmetry and positive semidefiniteness of $P_u$ give the exact spectral
gap $\kappa_u$ of the normalized infinite-temperature generator by
\Cref{thm:associated-gap}.
\Cref{lem:carre-L0} verifies \Cref{ass:poincare} with this constant,
so \Cref{thm:metropolis-all-temperature} applies.
The generator is independent of the Kraus representation by
\Cref{lem:relabelling}.
\end{proof}

\section{Spectral gap bounds and examples for Gibbs sampling with Hodge Laplacians}
\label{app:hodge-support}

We give further spectral gap bounds and examples for the Hodge Laplacian
Hamiltonians in \Cref{app:hodge-details}. The final example gives an upper
bound on the spectral gap of the CKG generator at low temperature for an
unweighted Hodge Laplacian.
Throughout, use the conventions of \Cref{sec:model,sec:walk}.

\subsection{Local spectral expansion}
\label{subsec:hodge-bounds}
\label{subsec:links}

Adopt \Cref{set:complex}.
For any vertex subset $\alpha$, write
$c_\alpha:=|\{f\in F:\alpha\subseteq f\}|$.
We use the induced link distributions and non-lazy link walks
of \cite[Secs.~2.2--2.3]{AlevLau2020}, with the uniform
distribution on $Y(k)=F$.

For $\alpha\in Y(j)$, $-1\le j\le k-2$, its link is
$\operatorname{lk}_Y(\alpha):=
\{\sigma\in Y:\sigma\cap\alpha=\emptyset,\ \sigma\cup\alpha\in Y\}$.
The \emph{link graph} $\Gamma_\alpha$ is the $1$-skeleton of
$\operatorname{lk}_Y(\alpha)$, regarded as a weighted graph.
Its vertex set is
$V_\alpha:=\{v\in[n]\setminus\alpha:
\alpha\cup\{v\}\in Y\}$.
Distinct vertices $u,v\in V_\alpha$ are adjacent precisely when
$\alpha\cup\{u,v\}\in Y$, and their edge has weight
$c_{\alpha\cup\{u,v\}}$.
The weighted degree of $u$ is
$(k-j-1)c_{\alpha\cup\{u\}}$.
Accordingly, the non-lazy random walk on $\Gamma_\alpha$ has
transition matrix $\sfM_\alpha$, defined by
$\sfM_\alpha(u,v):=
c_{\alpha\cup\{u,v\}}/
[(k-j-1)c_{\alpha\cup\{u\}}]$ for $u\ne v$ and
$\sfM_\alpha(u,u):=0$.

By the properties of induced link walks
\cite[Secs.~2.2--2.3]{AlevLau2020}, $\sfM_\alpha$ is
row-stochastic and reversible with respect to the probability
distribution
$\nu_\alpha(u):=
c_{\alpha\cup\{u\}}/((k-j)c_\alpha)$.
This distribution has full support.
Moreover, purity of $Y$ and $j\le k-2$ imply
$|V_\alpha|\ge2$.
Thus $\sfM_\alpha$ has real eigenvalues, and
$\lambda_2(\sfM_\alpha)$ denotes its second-largest eigenvalue,
counting multiplicity.

Define
$\mu_j(Y):=\max_{\alpha\in Y(j)}
\lambda_2(\sfM_\alpha)$ for $-1\le j\le k-2$, and set
\begin{equation}
 G_k(Y):=\frac1{k+1}
 \prod_{j=-1}^{k-2}\bigl(1-\mu_j(Y)\bigr).
 \label{eq:Gk}
\end{equation}
Following the one-sided convention of
\cite[Definition~1.1]{AlevLau2020}, for $0\le\gamma<1$
we call $Y$ a \emph{one-sided $\gamma$-local spectral expander}
if $\mu_j(Y)\le\gamma$ for every $-1\le j\le k-2$.
Only upper bounds on the second-largest eigenvalues are required;
the smallest eigenvalues need not be bounded away from $-1$.
Since $Y$ is finite, the condition $\mu_j(Y)<1$ for every
relevant $j$ is equivalent to the existence of such a $\gamma<1$.
For a family of complexes, a uniform expansion bound requires
a common $\gamma<1$.

The following statement specializes
\cite[Thm.~1.5, arXiv v3]{AlevLau2020}
to the uniform distribution on $Y(k)$.

\begin{theorem}[Alev--Lau bound for the uniform down--up walk~\cite{AlevLau2020}]
\label{thm:alev-lau}
Under \Cref{set:complex}, the uniform down--up walk $P_u$ satisfies
\begin{equation}
 1-\lambda_2(P_u)\ge G_k(Y).
 \label{eq:alev-lau}
\end{equation}
In particular, if $\mu_j(Y)<1$ for every $-1\le j\le k-2$,
then $P_u$ is irreducible and its spectral gap is at least
$G_k(Y)>0$.
\end{theorem}

\begin{proof}[Verification of the induced distributions]
The complex $Y=\langle F\rangle$ is pure of dimension $k$.
The uniform distribution $\Pi_k(f)=D_k^{-1}$ on $Y(k)$ induces
$\Pi_j(\alpha)=c_\alpha/[D_k\binom{k+1}{j+1}]$ for $-1\le j\le k$.
For $\alpha\in Y(j)$, $-1\le j\le k-2$, its link distributions are
\[
 \Pi_0^\alpha(u)=\frac{c_{\alpha\cup\{u\}}}{(k-j)c_\alpha},
 \qquad
 \Pi_1^\alpha(\{u,v\})=
 \frac{c_{\alpha\cup\{u,v\}}}{\binom{k-j}{2}c_\alpha}.
\]
Thus $\Pi_1^\alpha(\{u,v\})/(2\Pi_0^\alpha(u))=\sfM_\alpha(u,v)$
for $u\ne v$, so the local spectral parameters are $\mu_j(Y)$.
Conditional on a ridge $\eta$, the probability of selecting
$g\supseteq\eta$ in the up step is
$\Pi_k(g)/((k+1)\Pi_{k-1}(\eta))=1/d_\eta$; hence the induced
walk is $P_u$. Applying \cite[Thm.~1.5, arXiv v3]{AlevLau2020} gives
\eqref{eq:alev-lau}. Positivity of $G_k(Y)$ gives irreducibility by
\Cref{eq:gap-irreducible}.
\end{proof}

\begin{remark}[The case $k=1$]
\label{rem:k1}
For $k=1$, the product in \eqref{eq:Gk} contains only
$1-\mu_{-1}(Y)$, and $\sfM_\emptyset$ is the simple random
walk on the graph with vertex set $Y(0)$ and edge set $F$.
In this case, \eqref{eq:alev-lau} is an equality.

Indeed, define
$\sfD_1:\mathbb C^F\to\mathbb C^{Y(0)}$ by
$(\sfD_1a)(v):=
d_{\{v\}}^{-1}\sum_{e\in F:\,v\in e}a_e$,
and define
$\sfU_0:\mathbb C^{Y(0)}\to\mathbb C^F$ by
$(\sfU_0h)(\{u,v\}):=\frac12(h(u)+h(v))$.
Then $P_u=\sfU_0\sfD_1$ and
$\sfD_1\sfU_0=\frac12(I+\sfM_\emptyset)$ have the same nonzero spectra
\cite[Fact~2.1]{AlevLau2020}. Both spectra are nonnegative, and the
independent unsigned incidence columns of two distinct edges give rank
at least two. Thus $\lambda_2(P_u)=\frac12(1+\mu_{-1}(Y))$ and
$1-\lambda_2(P_u)=G_1(Y)$.
\end{remark}

\begin{corollary}[Exact spectral gap of the normalized infinite-temperature generator for an irreducible down--up walk]
\label{cor:Pu-gap}
Adopt \Cref{set:complex} and suppose $P_u$ is irreducible. For the
associated normalized infinite-temperature generator on $\ell^2(X(k))$, one has
$\gap(-\cL_0)=\kappa_u:=1-\lambda_2(P_u)>0$ and
$\cE_0(A)\ge\kappa_u\Var_\tau(A)$ for every operator $A$.
If, in addition, $Y$ satisfies one-sided local spectral expansion,
then $\kappa_u\ge G_k(Y)>0$.
\end{corollary}
\begin{proof}
By \Cref{lem:Pu-structure}, $P_u$ is symmetric and positive semidefinite.
Hence \Cref{thm:associated-gap} gives
$\gap(-\cL_0)=\min\{1,1-\lambda_2(P_u)\}=1-\lambda_2(P_u)$.
The final conditional assertion is \Cref{thm:alev-lau}.
\end{proof}

\subsubsection{Consequence for the finite-temperature spectral gap of the CKG generator}

\begin{corollary}[Spectral gap bound from local spectral expansion]
\label{cor:hodge-local-expansion}
Adopt \Cref{set:complex}. If $\mu_j(Y)<1$ for every
$-1\le j\le k-2$, then $P_u$ is irreducible and
\begin{equation}
 \kappa_u\ge G_k(Y)>0.
 \label{eq:main-link-gap}
\end{equation}
For the Hamiltonian and couplings of \Cref{thm:main},
\begin{equation}
 \gap(-\cL_{\gamma_{\rm M}})\ge\frac12e^{-\beta W}G_k(Y)
 \qquad(\beta>0).
 \label{eq:hodge-link-metropolis}
\end{equation}
\end{corollary}

\begin{proof}
Apply \Cref{thm:alev-lau,thm:main}.
\end{proof}

\subsection{Perturbative bounds and bounded-degree consequences}
\label{app:hodge-consequences}

\begin{corollary}[Perturbative spectral gap bounds for CKG generators with Hodge Laplacians]
\label{cor:hodge-perturbative}
Under the Hamiltonian, coupling, and irreducibility assumptions of
\Cref{thm:main}, take $\beta>0$, put $b:=\beta W$, and choose
$\gamma\in\{\gamma_{\rm G},\gamma_{\rm M}\}$.
Put $\delta_{\gamma,u}:=\sqrt2e^b b/
\sqrt{q_\gamma\kappa_u}$.
If $\delta_{\gamma,u}<1$, the CKG generator $\cL_\gamma$ is primitive with unique
stationary state $\sigma_\beta$ and
\begin{equation}
 \gap(-\cL_\gamma)\ge
 q_\gamma(1-\delta_{\gamma,u})^2e^{-\beta W}\kappa_u.
 \label{eq:hodge-direct-perturbative}
\end{equation}
If $G_k(Y)>0$, put
$\delta_{\gamma,Y}:=\sqrt2e^b b/\sqrt{q_\gamma G_k(Y)}$.
When $\delta_{\gamma,Y}<1$, the same conclusion holds with $\kappa_u$
and $\delta_{\gamma,u}$ replaced by $G_k(Y)$ and $\delta_{\gamma,Y}$,
respectively. In particular,
\begin{equation}
 \beta W\le\frac{\sqrt{q_\gamma G_k(Y)}}{2\sqrt2e}
 \quad\Longrightarrow\quad
 \gap(-\cL_\gamma)\ge\frac{q_\gamma}{4}e^{-\beta W}G_k(Y).
 \label{eq:simple-gap}
\end{equation}
\end{corollary}

\begin{proof}
\Cref{cor:Pu-gap,lem:carre-L0} verify \Cref{ass:poincare} with
$\kappa_0=\kappa_u$, so \Cref{thm:comparison} gives
\eqref{eq:hodge-direct-perturbative}. If $G_k(Y)>0$, then
\Cref{thm:alev-lau} allows the smaller constant $\kappa_0=G_k(Y)$.
For the final implication, $G_k(Y)\le\kappa_u\le1$ and $q_\gamma\le1$
imply $\beta W<1$. Hence $e^{\beta W}\le e$ and
$\delta_{\gamma,Y}\le1/2$, which proves \eqref{eq:simple-gap}.
\end{proof}

\begin{remark}[Irreducibility does not require local spectral expansion]
\label{rem:links-not-necessary}
Consider the pure complex generated by the triangles
$\{1,2,3\}$, $\{2,3,4\}$, $\{3,4,5\}$, and $\{1,4,5\}$.
Only consecutive triangles in this list share an edge. Each shared edge
has degree two, and each other edge has degree one. Thus
$P_u=I-L_{\mathrm{path},4}/6$, where $L_{\mathrm{path},4}$ is the
combinatorial Laplacian of the path on four vertices, and
$\kappa_u=(2-\sqrt2)/6>0$. The link of vertex $1$ consists of the
two disjoint edges $\{2,3\}$ and $\{4,5\}$, so $\mu_0(Y)=1$ and
$G_2(Y)=0$. \Cref{thm:main} gives a positive lower bound on the CKG Lindbladian spectral gap although \Cref{cor:hodge-local-expansion} does not apply. This example satisfies \Cref{set:complex}; it is not a clique complex.
\end{remark}

\begin{corollary}[Bounded-degree complexes with one-sided local spectral expansion]
\label{cor:bounded-degree}
Fix $k\ge1$, $0\le\bar\mu<1$, and $R_-,R_+<\infty$.  Suppose every member
$X$ of a family of finite simplicial complexes satisfies \Cref{set:complex}
together with the degree bounds \eqref{eq:degree-bounds} and
$\mu_j(Y)\le\bar\mu$ for $-1\le j\le k-2$, where
$Y=\langle X(k)\rangle$.
Set $C_-:=(k+1)R_-$, $C_+:=(k+2)R_+$, $\Lambda_k:=C_-+C_+$,
$C_H:=\max\{C_-,C_+\}$, and $\kappa_*:=(1-\bar\mu)^k/(k+1)$.
Fix either transition weight $\gamma\in\{\gamma_{\rm G},\gamma_{\rm M}\}$
and set $B_*:=\sqrt{q_\gamma\kappa_*}/(2\sqrt2e)$ and
$g_*:=q_\gamma e^{-B_*}\kappa_*/4>0$.
For inverse temperature $\theta$ with respect to the unscaled Hodge Laplacian $\Delta_k$, use
$\alpha_k=\sqrt{C_-}$, $\alpha_{k+1}=\sqrt{C_+}$, $H_{X,k}=\Delta_k/\Lambda_k$,
and $\beta=\theta\Lambda_k$.  If $0<\theta C_H\le B_*$, then the exact CKG
generator $\cL_\gamma$, constructed from the Hamiltonian and couplings of
\Cref{thm:main} using the selected transition weight $\gamma$, has the unique stationary state
$e^{-\theta\Delta_k}/\Tr(e^{-\theta\Delta_k})$ and
\begin{equation}
 \gap(-\cL_\gamma)\ge
 \frac{q_\gamma}{4}e^{-\theta C_H}\kappa_*\ \ge\ g_*.     \label{eq:constant-gap}
\end{equation}
The constants $B_*$ and $g_*$ depend only on $k$, $\bar\mu$, and the chosen
weight, so \eqref{eq:constant-gap} is a positive lower bound uniform over
the family and independent of $|X(k)|$.  The admissible interval $0<\theta\le B_*/C_H$ is nonempty and the same for every complex in the family.
For $\gamma=\gamma_{\rm M}$, the following bound holds for every
$\theta>0$:
\begin{equation}
 \gap(-\cL_{\gamma_{\rm M}})
 \ge\frac12e^{-\theta C_H}\kappa_* .                    \label{eq:bounded-degree-all-temperature-gap}
\end{equation}
Consequently, for any fixed $B>0$, the condition $\theta C_H\le B$ gives
the following bound, uniform over the family
\begin{equation}
 \gap(-\cL_{\gamma_{\rm M}})
 \ge\frac12e^{-B}\kappa_* .                            \label{eq:bounded-degree-fixed-B-gap}
\end{equation}
This bound does not require $B\le B_*$.
\end{corollary}

\begin{proof}
The product in \eqref{eq:Gk} gives $G_k(Y)\ge\kappa_*$, and
\Cref{lem:boundary-norms} gives $\|\Delta_k\|\le C_H$, hence
$\beta W\le\theta C_H$. Substituting these bounds into
\eqref{eq:simple-gap} gives \eqref{eq:constant-gap} whenever
$\theta C_H\le B_*$. For $\gamma=\gamma_{\rm M}$, the same substitutions
in \eqref{eq:hodge-link-metropolis} give
\eqref{eq:bounded-degree-all-temperature-gap} and
\eqref{eq:bounded-degree-fixed-B-gap}.
Finally $C_H\ge C_-\ge k+1>0$ and $\kappa_*>0$, so the interval
$0<\theta C_H\le B_*$ is nonempty.
\end{proof}

\begin{corollary}[Mixing-time upper bound for bounded-degree families]
\label{cor:HDX-mixing}
Under the hypotheses of \Cref{cor:bounded-degree}, with $B_*$ and $g_*$
defined there, for every $0<\epsilon<1$,
\[
 t_{\rm mix}(\epsilon)
 \le\frac1{g_*}\Bigl(
 \log\frac1\epsilon+\frac12\log D_k+\frac{B_*}2\Bigr),
 \qquad\text{hence}\qquad
 t_{\rm mix}(\epsilon)=O(\log D_k+\log\epsilon^{-1}),
\]
where the $O(\cdot)$ constant depends only on $k$, $\bar\mu$, and the chosen
weight, and not on the member of the family, on $D_k$, or on $\epsilon$.
For $\gamma=\gamma_{\rm M}$, let $B>0$ be arbitrary and assume only
$\theta C_H\le B$.  Then
\[
 t_{\rm mix}(\epsilon)
 \le\frac{2e^B}{\kappa_*}
 \left(\log\frac1\epsilon+\frac12\log D_k+\frac B2\right).
\]
At fixed accuracy, every fixed value of $B$ gives a mixing-time upper
bound logarithmic in $D_k$; the prefactor grows exponentially with $B$.
\end{corollary}

\begin{proof}
Apply \Cref{lem:mixing} with $D=D_k$, $g=g_*$, and
$\beta W\le B_*$, using the mixing-time upper bound
$g^{-1}[\log(1/\epsilon)+(\log D_k+\beta W)/2]$.
For the bound with $\gamma=\gamma_{\rm M}$, use $g=\frac12e^{-B}\kappa_*$ and $\beta W\le B$.
The stated dependence of the constants follows from \Cref{cor:bounded-degree}.
\end{proof}

\subsection{Approximation of the harmonic state}
\label{app:harmonic-preparation}

The following corollaries combine the Gibbs-state error relative to the
harmonic state with the mixing time bounds already proved.

\begin{corollary}[Trace-distance bound for harmonic state preparation]
\label{cor:harmonic-approximation}
Under the assumptions of \Cref{thm:main}, assume $0<h:=\dim\ker\Delta_k<D_k$.
Let $P_0$ project onto this kernel, set $R=(D_k-h)/h$, and let
$\delta_H$ be the smallest positive eigenvalue of
$H_{X,k}=\Delta_k/\Lambda_k$.
Write $\rho_t=e^{t\cL_{\gamma_{\rm M},*}}(\rho_0)$ for the state evolved
from any density operator $\rho_0$ under the CKG dynamics with transition weight $\gamma_{\rm M}$.
For $\beta>0$ and $0<\eta<1$,
\[
 t\ge\frac{2e^{\beta W}}{\kappa_u}
 \left[\log\frac1\eta+\frac12\log D_k+\frac{\beta W}{2}\right]
 \quad\Longrightarrow\quad
 \frac12\|\rho_t-P_0/h\|_1\le\eta+R e^{-\beta\delta_H}.
\]
\end{corollary}
\begin{proof}
Put $Q_\beta=\Tr[(I-P_0)e^{-\beta H_{X,k}}]
\le(D_k-h)e^{-\beta\delta_H}$. The ground-space and positive-energy parts of $\sigma_\beta$ have
orthogonal supports, so
$\tfrac12\|\sigma_\beta-P_0/h\|_1=Q_\beta/(h+Q_\beta)
\le R e^{-\beta\delta_H}$.
By \Cref{thm:main,lem:mixing}, the stated evolution time ensures
$\tfrac12\|\rho_t-\sigma_\beta\|_1\le\eta$; the triangle inequality
proves the claim.
\end{proof}

\begin{corollary}[Polynomial evolution time for harmonic state preparation]
\label{cor:harmonic-polynomial}
Consider a family satisfying \Cref{cor:harmonic-approximation}, indexed
by $n\ge2$, with
$\kappa_u\ge n^{-a}$, $R\le n^r$, and $\delta_H/W\ge\alpha>0$
for constants $a,r\ge0$ and $\alpha$ independent of $n$.
For fixed $s>0$, choose a constant $c\ge(r+s)/\alpha$ and
$\beta W=c\log n$. Then, for every initial state,
\[
 \frac12\|\rho_t-P_0/h\|_1\le2n^{-s}
 \quad\text{whenever}\quad
 t\ge2n^{a+c}\left[\left(s+\frac c2\right)\log n
                       +\frac12\log D_k\right].
\]
The displayed sufficient evolution time is polynomial in $n$ if $\log D_k$ is polynomially bounded in $n$;
for complexes on $n$ vertices, $D_k\le2^n$ ensures this condition.
\end{corollary}
\begin{proof}
The Gibbs-state error relative to the harmonic state obeys
$R e^{-\beta\delta_H}\le n^{r-c\alpha}\le n^{-s}$.
Apply \Cref{cor:harmonic-approximation} with $\eta=n^{-s}$,
$\kappa_u^{-1}\le n^a$, and $e^{\beta W}=n^c$.
\end{proof}

The ratio $\delta_H/W$ is unchanged by multiplying the Hamiltonian by a positive scalar.
These sufficient conditions concern evolution time; implementation and
measurement costs remain separate requirements.

\subsection{An example with reducible projected couplings}
\label{app:hodge-barbell-example}

The following example fixes one coupling family and applies
\Cref{appth:thm:groundobstruction} to show that the CKG Lindbladian spectral gap
tends to zero as the inverse temperature tends to infinity, while the
down--up walk has a fixed positive spectral gap.
Here $\cL_\beta$ denotes the CKG
generator with transition weight $\gamma_{\rm M}$ for
the Hamiltonian and couplings in the corollary.

\begin{corollary}[Low-temperature CKG Lindbladian spectral gap bound for an unweighted Hodge Laplacian]
\label{appth:cor:barbell}
Let $G$ consist of two disjoint four-cycles joined by one bridge, let
$X=\mathrm{Cl}(G)$, and put $F=E(G)$ and $H=\Delta_1(X)$.
The uniform edge down--up walk has transition matrix $Q=P_u$, with
$Q(f,g)=\frac12\sum_{v\in f\cap g}1/\deg_G(v)$, where $\deg_G(v)$ is
the number of edges incident to $v$. For its couplings
$K_{fg}=\sqrt{Q(f,g)}\ket g\bra f$,
\begin{equation}
 \kappa_0=\kappa_u=\frac{4-\sqrt{10}}{12}>0,
 \qquad \gap(-\cL_\beta)\le2e^{-\beta\delta_H},
 \label{appth:eq:barbell}
\end{equation}
where $\delta_H$ is the smallest positive eigenvalue of $H=\Delta_1(X)$.
For a rescaled Hamiltonian $H'=\Delta_1(X)/\Lambda$, with fixed
$\Lambda>0$, the corresponding generator at inverse temperature $\beta$
instead satisfies $\gap(-\cL^{H'}_\beta)\le
2e^{-\beta\delta_H/\Lambda}$, with the same spectral gap of the normalized infinite-temperature generator.
\end{corollary}

\begin{proof}
The graph has eight vertices, nine edges, and no triangles. Thus
$H=B^\dagger B$ for its oriented incidence matrix $B$, and its kernel
has dimension two. The normalized oriented cycle vectors $h_L,h_R$, supported on
the respective four-cycles, form an orthonormal basis of this kernel for
any fixed orientation convention. The projection
$P=\ket{h_L}\bra{h_L}$ commutes with every coupling operator
projected to the ground state subspace: a coupling $K_{fg}$ connects an edge only to itself or to an edge
sharing a vertex, and
no edge of the left cycle is incident to an edge of the right cycle.
Every matrix element of a projected coupling between $h_L$ and $h_R$
therefore vanishes. \Cref{appth:thm:groundobstruction} gives the upper
bound in \eqref{appth:eq:barbell}.

To compute the spectral gap of the normalized infinite-temperature generator, let $M_G$ be the vertex random walk, with
$M_G(v,w)=1/\deg_G(v)$ for adjacent vertices and zero otherwise.
The nonzero eigenvalues of the edge down--up walk coincide, with
multiplicity, with the nonzero eigenvalues of $(I+M_G)/2$.
Antisymmetric functions on the two neighbors of the attachment vertex in
each cycle give two zero eigenvalues of $M_G$. On the remaining subspace,
the subspaces symmetric and antisymmetric under exchange of the cycles
reduce $M_G$ to the matrices
\[
 \begin{pmatrix}\pm1/3&2/3&0\\1/2&0&1/2\\0&1&0\end{pmatrix}.
\]
Their combined eigenvalues are $\pm1$,
$\pm(2+\sqrt{10})/6$, and $\pm(\sqrt{10}-2)/6$.
Consequently the uniform edge down--up walk has spectral gap
$(4-\sqrt{10})/12$. Its associated normalized infinite-temperature
generator has the same spectral gap by \Cref{thm:associated-gap}.
The incidence matrix of this connected graph has rank seven, so the
Hodge Laplacian has seven positive eigenvalues.
Finally, dividing the Hamiltonian by $\Lambda$ leaves its ground state subspace
unchanged and divides its smallest positive eigenvalue by $\Lambda$. Applying
\Cref{appth:thm:groundobstruction} once more proves the rescaled statement.
\end{proof}

\section{Relation to previous comparison theorems}
\label{app:comparison-context}

We distinguish the comparison proved here from established
comparison inequalities for different stationary states and from results for other choices of CKG
couplings.

\subsection{Comparison under GNS detailed balance}

Junge, LaRacuente, and Rouz\'e prove noncommutative Holley--Stroock
comparisons for modified logarithmic Sobolev inequalities, including
complete versions \cite[Theorems~1.2 and~3.1]{JLR}. Their generators are
GNS symmetric; the comparison generator at the maximally mixed state retains the modular
eigenoperators and removes their Bohr-frequency weights
\cite[Sec.~2, Eqs.~(2.2)--(2.3)]{JLR}.

Li and Lu compare primitive GNS-symmetric generators with commuting
full-rank stationary states and the same family of modular eigenoperators
\cite[Theorem~3.27, Eqs.~(3.50)--(3.54)]{LiLu}. At $p=2$, their Beckner
constant is half the spectral gap
\cite[Lemma~3.18 and Proposition~3.19]{LiLu}. In particular, their
specialization to $\tau=I/D$ gives
\begin{equation}
 \gap(-\mathcal G_\sigma)
 \ge \frac{\lambda_{\min}(\sigma)}{\lambda_{\max}(\sigma)}
       \gap(-\mathcal G_\tau)
 =e^{-\beta W}\gap(-\mathcal G_\tau)
 \quad\text{when }\sigma=\sigma_\beta,
 \label{eq:prior-gns-comparison}
\end{equation}
for the generator pair specified there
\cite[Remark~3.28]{LiLu}. The generator pairs in the cited GNS comparisons
differ from the pair considered here. Our theorem does not imply those comparisons or improve
their numerical prefactor. The additional estimate in
\Cref{thm:metropolis-all-temperature} bounds the integral of squared commutator norms in \eqref{eq:IM-definition} for the given CKG couplings without assuming GNS symmetry.

The difference in hypotheses occurs even on a qubit. Recall that a
generator is \emph{modular covariant} with respect to $\sigma$ if it
commutes with every map $O\mapsto\sigma^{is}O\sigma^{-is}$,
$s\in\mathbb R$. GNS symmetry implies this covariance
\cite[Lemma~2.2]{LiLu}; for a Gibbs state with $\beta>0$, it is equivalent
to commutation with $O\mapsto e^{itH}Oe^{-itH}$ for every real $t$.

\begin{proposition}[A CKG generator outside the GNS comparison hypotheses]
\label{prop:non-GNS-example}
On $\mathbb C^2$, let
$X=|0\rangle\langle1|+|1\rangle\langle0|$ and
$Z=|0\rangle\langle0|-|1\rangle\langle1|$. Fix $E>0$ and take
\begin{equation}
 H=E|1\rangle\langle1|,\qquad
 A_1=\frac{X+Z}{2},\qquad A_2=\frac{Z}{\sqrt2}.
 \label{eq:non-GNS-family}
\end{equation}
The normalized infinite-temperature generator for this Hermitian family has
spectral gap $\kappa_0=1-1/\sqrt2$. For every $\beta>0$, its exact CKG generator with transition weight $\gamma_{\rm M}$ is primitive and KMS symmetric, with
$\gap(-\cL_{\gamma_{\rm M}})\ge
\frac12e^{-\beta E}(1-1/\sqrt2)$, but it is not modular covariant and
therefore is not GNS symmetric.
\end{proposition}

\begin{proof}
The identities $X^2=Z^2=I$ and $XZ+ZX=0$ give
$A_1^2+A_2^2=I$. Put $Y=iXZ$. The coupling channel
$\Phi(O)=A_1OA_1+A_2OA_2$ satisfies
$\Phi(X)=(-X+Z)/2$, $\Phi(Z)=(X+Z)/2$, and $\Phi(Y)=-Y$.
Its eigenvalues on the traceless operators are
$1/\sqrt2,-1/\sqrt2,-1$. Hence the normalized infinite-temperature generator $\Phi-\id$ has the
asserted spectral gap.
\Cref{thm:metropolis-all-temperature} gives the spectral gap
lower bound at finite temperature and primitivity; \eqref{eq:KMS-DB} gives KMS symmetry.

Write $\varphi(\omega)=\widehat f_\beta(\omega)
=\sqrt\beta(2\pi)^{-1/4}e^{-\beta^2\omega^2/4}$, set
$t_E=\tanh(\beta E/4)$, and define the positive numbers
$w_\pm:=\int_{\mathbb R}\gamma_{\rm M}(\omega)
\varphi(\omega)\varphi(\omega\mp E)\,d\omega$.
The filtered first coupling is
\[
 \widehat A_1(\omega)=\frac12
 \begin{pmatrix}
  \varphi(\omega)&\varphi(\omega+E)\\
  \varphi(\omega-E)&-\varphi(\omega)
 \end{pmatrix},
\]
while $\widehat A_2(\omega)=\varphi(\omega)Z/\sqrt2$ contributes zero
to $\cL_{\gamma_{\rm M}}(Z)$. In the conventions of
\eqref{eq:CKG-data}, the loss operator and coherent correction satisfy
$(R_{\gamma_{\rm M}})_{01}=(w_--w_+)/4$ and
$(B_{\gamma_{\rm M}})_{01}=-it_E(R_{\gamma_{\rm M}})_{01}/2$.
The anticommutator with $Z$ has zero $(0,1)$ entry. Consequently,
\begin{equation}
 \bigl(\cL_{\gamma_{\rm M}}(Z)\bigr)_{01}
 =\frac{(1-t_E)w_-+(1+t_E)w_+}{4}>0.
 \label{eq:non-GNS-witness}
\end{equation}
Since $[H,Z]=0$ but $[H,\cL_{\gamma_{\rm M}}(Z)]\ne0$, the generator
does not commute with Hamiltonian conjugation. The cited implication
from GNS symmetry to modular covariance proves the last assertion.
\end{proof}

\begin{remark}[Failure of GNS symmetry for a construction from a symmetric Markov chain]
\label{rem:non-GNS-classical-reference}
Put $V=(X+Z)/\sqrt2$ and consider the normalized Hermitian family
$\{V/2,\allowbreak\sqrt{3/8}I,\allowbreak X/\sqrt8,\allowbreak Y/\sqrt8,\allowbreak Z/\sqrt8\}$.
Its channel is
$\Phi=\frac12\Phi_{\mathrm{reset}}+
\frac14\id+\frac14\operatorname{Ad}_V
=\frac12\Phi_{\mathrm{reset}}+\frac12\mathcal D_V$,
where $\Phi_{\mathrm{reset}}(O)=\Tr(O)I/2$,
$\operatorname{Ad}_V(O)=VOV$, and $\mathcal D_V$ dephases in the
eigenbasis of $V$. In that basis it is precisely the associated channel
of $Q=\left(\begin{smallmatrix}3/4&1/4\\1/4&3/4\end{smallmatrix}\right)$.
The normalized infinite-temperature generator has spectral gap $1/2$. For the same Hamiltonian $H=E|1\rangle\langle1|$,
the $I,X,Y,Z$ couplings have zero off-diagonal contribution to
$\cL_{\gamma_{\rm M}}(Z)$. The $V$ coupling gives the strictly positive
right side of \eqref{eq:non-GNS-witness} divided by two. Thus even a
channel associated with a symmetric Markov chain can yield a CKG generator outside
the GNS comparison hypotheses when the Hamiltonian is not diagonal in the basis used to construct the channel.
\end{remark}

\subsection{General comparison principles and local high-temperature bounds}

Junge and Wu characterize complete domination of operator-valued
commutator forms, called Lindblad carr\'es du champ, by inclusion of the
linear spans of the jump operators together
with the identity
\cite[Theorem~3.1]{JW}, without a detailed-balance assumption. Their
Theorem~6.1 separately gives stability of the spectral gap under small perturbations between generators that are GNS symmetric with respect to the same stationary state and have the same fixed-point algebra. The domination criterion determines whether a domination constant exists;
it does not by itself give the coefficient in
\eqref{eq:metropolis-all-temperature-gap}.

Bergamaschi compares the Dirichlet form of the CKG generator with a Gibbs-weighted commutator form
for single-site Pauli couplings under a high-temperature condition
controlled by local interaction strength
\cite[Theorems~1.3 and~5.1]{Bergamaschi}. Related locality-based
thermalization bounds appear in \cite{RouzeFrancaAlhambra2024}.

For Hodge Laplacians, the spectral gap bounds for the normalized
infinite-temperature generator follow by transferring the classical
walk bound of \cite{AlevLau2020}; they are not an independent quantum
local-to-global result.
All gap comparisons account for the overall normalization of the generators.

\subsection{Random coupling constructions}

Ramkumar and Soleimanifar obtain a constant spectral gap of the CKG generator, with fixed success
probability, from $M=\Theta(\log D)$ random unitary $1$-design couplings
when $\beta\|H\|=O(1)$, with filter width $\sigma_E=\beta^{-1}$
\cite[Theorem~2.2, arXiv v1]{RamkumarSoleimanifar}. After shifting the energy origin, this condition is equivalent to
$\beta W=O(1)$. Their Theorem~2.3
gives a guarantee depending on the distribution of energy eigenvalues beyond that regime.
Their results sample the coupling family at random; our theorem assumes
a given family and uses its infinite-temperature Poincar\'e constant.
Neither theorem gives a uniformly stronger bound under the hypotheses
of the other.

\bibliographystyle{unsrtnat}
\bibliography{main}
\end{document}